\documentclass{lmcs} %%% last changed 2014-08-20
\pdfoutput=1

\usepackage[utf8]{inputenc}

\usepackage{lastpage}
\lmcsdoi{19}{1}{4}
\lmcsheading{}{\pageref{LastPage}}{}{}%
{Sep.~03,~2021}{Jan.~16,~2023}{}


\usepackage{amssymb}
\usepackage{amsmath}
\usepackage{stmaryrd}
\usepackage{adjustbox}
\usepackage{subcaption}
\usepackage{multicol}

\usepackage{hyperref}
\theoremstyle{plain} %\crefname{satz}{Satz}{S\"atze}
\usepackage{array}

\newcommand{\be}{\begin{enumerate}}
\newcommand{\ee}{\end{enumerate}}
\newcommand{\bi}{\begin{itemize}}
\newcommand{\ei}{\end{itemize}}

\newcommand\etal{\textit{et al.}}
\newcommand\ie{\textit{i.e.}}

\newcommand{\ra}{\rightarrow}
\newcommand{\la}{\leftarrow}
\newcommand{\Ra}{\Rightarrow}

\newcommand{\lra}{\longrightarrow}

\newcommand{\cA}{\mathcal{A}}

\newcommand{\cC}{\mathcal{C}}

\newcommand{\cF}{\mathcal{F}}

\newcommand{\cL}{\mathcal{L}}
\newcommand{\cM}{\mathcal{M}}

\newcommand{\cP}{\mathcal{P}}

\newcommand{\cW}{\mathcal{W}}

\newcommand{\N}{\mathbb{N}}
\newcommand{\R}{\mathbb{R}}

\newcommand{\Z}{\mathbb{Z}}

\newcommand{\trans}[3]{#1 \stackrel{#2}{\lra} #3}

\newcommand{\myquad}[1][1]{\hspace*{#1em}\ignorespaces}

\usetikzlibrary{arrows,arrows.meta,automata,positioning,quotes}
\usetikzlibrary{cd} % for commuting diagrams
\usetikzlibrary{shapes,fit,calc}
\tikzstyle{smallblock} = [draw, fill=white, rectangle, 
    minimum height=2cm, minimum width=3cm]
\tikzstyle{block} = [draw, fill=white, rectangle, 
    minimum height=2cm, minimum width=4cm]
\tikzstyle{bigblock} = [draw, fill=white, rectangle, 
    minimum height=4cm, minimum width=8cm]
\tikzstyle{input} = [coordinate]
\tikzstyle{output} = [coordinate]
\tikzstyle{pinstyle} = [pin edge={to-,thin,black}]

\tikzset{
->, % makes the edges directed
>=stealth', % makes the arrow heads bold
node distance=1.0cm, % specifies the minimum distance between two nodes. Change if necessary.
every state/.style={thick}, % sets the properties for each 'state' node
initial text=$ $, % sets the text that appears on the start arrow
}

\begin{document}

\title[Coalgebras for Bisimulation]{Coalgebras for Bisimulation of\texorpdfstring{\\}{ }Weighted Automata over Semirings}
\author[P.~Bhaduri]{Purandar Bhaduri\lmcsorcid{0000-0002-8847-0394}}
\address{Indian Institute of Technology Guwahati\\ Guwahati 781039, India}
\email{pbhaduri@iitg.ac.in}
\sloppy

\begin{abstract}
Weighted automata are a generalization of nondeterministic automata that
associate a weight drawn from a semiring $K$ with every transition and every
state. Their behaviours can be formalized either as weighted language
equivalence or weighted bisimulation. In this paper we explore the properties of
weighted automata in the framework of coalgebras over (i) the category
$\mathsf{SMod}$ of semimodules over a semiring $K$ and $K$-linear maps, and (ii)
the category $\mathsf{Set}$ of sets and maps. We show that the behavioural
equivalences defined by the corresponding final coalgebras
in these two cases characterize weighted language equivalence and weighted
bisimulation, respectively. These results extend earlier work by Bonchi \etal\
using the category $\mathsf{Vect}$ of vector spaces and linear maps as the
underlying model for weighted automata with weights drawn from a field $K$. The
key step in our work is generalizing the notions of linear relation and linear
bisimulation of Boreale from vector spaces to semimodules using the concept of
the kernel of a $K$-linear map in the sense of universal algebra. We also
provide an abstract procedure for forward partition refinement for computing
weighted language equivalence. Since for weighted automata defined over
semirings the problem is undecidable in general, it is guaranteed to halt only
in special cases. We provide sufficient conditions for the termination
of our procedure. Although the results are similar to those of Bonchi \etal,
many of our proofs are new, especially those about the coalgebra in $\mathsf{SMod}$
characterizing weighted language equivalence.
\end{abstract}

\maketitle

\section{Introduction}
\label{sec:intro}

Bisimulation was introduced by Park and
Milner~\cite{park1981concurrency,milner1989communication} for characterizing the
equivalence of two processes specified by transition systems or process algebra
terms. Over the years the notion has had an enduring impact on the study of the
behaviour of systems, with ramifications in concurrency, automata theory, modal
logic, coalgebras, games and formal verification. The basic notion of
bisimulation for discrete systems has been extended to probabilistic,
quantitative and even continuous systems.

Coalgebras are a category theoretic concept that enable a study of state
transition systems and their behaviours in a unified
setting~\cite{rutten2000universal,jacobs2016introduction}. The idea behind a
coalgebraic theory of systems is the following. Given an endofunctor $\cF :\cC
\ra \cC$ on a concrete category $\cC$, a coalgebra $f: X \ra \cF X$ represents a
transition system on the set of states $X$ possibly with additional structure.
For example, in $\mathsf{Set}$, the category of sets and maps, a coalgebra for
the endofunctor $\cF_d X = 2 \times X^A$ represents a deterministic automaton with
states $X$ and alphabet $A$. This is because a coalgebra for $\cF_d$ can be seen
as a map $\langle o, \delta \rangle : X \ra 2 \times X^A$ where $o: X \ra 2$ is
the output map indicating whether a given state is accepting and $\delta: X \ra
X^A$ is the transition map where $\delta(x)(a)$ is the next state on reading the
letter $a$ in state $x$. Here $2=\{0,1\}$ is the set of truth values. Similarly
a coalgebra for the endofunctor $\cF_n(X) = 2 \times (\cP_\omega X)^A$ on
$\mathsf{Set}$, where $\cP_\omega X$ is the finite powerset of $X$, represents a
nondeterministic automaton.

Homomorphisms between coalgebras can be seen as behaviour-preserving maps. A
final coalgebra for the functor $\cF: \cC \ra \cC$, if it exists, is the
universe of all possible $\cF$-behaviours. The unique arrow from any coalgebra
to the final coalgebra maps a state to its behaviour. The final coalgebra for
the functor $\cF$ naturally induces a notion of $\cF$-behavioural equivalence,
denoted by $\approx_\cF$. Two states are behaviourally equivalent if they are
mapped to the same element in the final coalgebra by the unique arrow. For
example, for deterministic automata, the behaviour of a state is the language
accepted by the automaton starting from that state; two states $x$ and $y$
satisfy $x \approx_{\cF_d}$ if and only if they are language equivalent. For
nondeterministic automata, two states are $\cF_n$-behaviourally equivalent for
the functor $\cF_n$ on $\mathsf{Set}$ defined above if and only if they are
bisimilar. However, if we consider the category $\mathsf{Rel}$ of sets and
relations, the behavioural equivalence for a suitable functor coincides with
language equivalence~\cite{hasuo2007generic} of nondeterministic automata.
Thus the notion of behavioural equivalence for the same computational object is
relative to the underlying category and the associated endofunctor.

In this paper we focus on bisimulation for weighted
automata~\cite{buchholz2008bisimulation}. Weighted automata were introduced by
Sh\"{u}tzenburger~\cite{schutzenberger1961definition} and have found renewed
interest in the last two decades~\cite{droste2009handbook} as they arise in
various contexts where quantitative modelling is involved. Intuitively, weighted
automata generalize the notion of nondeterministic automata where each state and
each transition has an associated weight valued in some semiring. Such a weight could
represent a cost, reward or probability, or any other measure of interest. Like
nondeterministic automata, the behaviour of weighted automata have two different
characterizations, in terms of weighted language equivalence and weighted
bisimulation.

In \cite{bonchi2012coalgebraic} Bonchi \etal\ have given a comprehensive account
of both weighted-language equivalence and bisimilarity of weighted automata in
terms of coalgebras of an endofunctor $\cL$ on $\mathsf{Vect}$ (the category of
vector spaces and linear maps) and an endofunctor $\cW$ on $\mathsf{Set}$,
respectively. The characterization of weighted language equivalence, is based on
an elegant notion of linear bisimulation given by
Boreale~\cite{boreale2009weighted}. The idea here is that a weighted automaton
is a vector space (of states) over a field $K$, with two linear maps: the output
map from states to observations valued in $K$ and an $A$-indexed family of
transition maps from states to states, where $A$ is the set of actions. Then, a
\emph{linear relation} over states is a set of pairs whose differences form a
subspace of the vector space and a linear bisimulation is a linear relation that
preserves the output map and (the corresponding subspace) is invariant under the
transition maps. Boreale showed that this notion of linear bisimulation
coincides with weighted language equivalence. Bonchi \etal\
~\cite{bonchi2012coalgebraic} gave a coalgebraic formulation of these results
using an endofunctor $\cL$ on $\mathsf{Vect}$.

Weighted automata were originally defined over semirings and not fields.
Therefore, it is desirable that the results of
Boreale~\cite{boreale2009weighted} and Bonchi \etal
~\cite{bonchi2012coalgebraic} related to linear bisimulation and language
equivalence be generalized to semimodules over semirings. This is exactly what
we accomplish in this paper, starting with the assumption that $K$ is a
semiring. Bonchi \etal\ had identified in their work the results which could be
extended to semirings in a straightforward way, and the ones which could not,
because the latter involve the minus operator in fields. These are the
the results of Boreale on linear bisimulation and their coalgebraic formulation
which make essential use of the properties of fields and vector spaces.
We show that there is an elegant characterization of these concepts in the
semiring-semimodule setting, by leveraging the concept of \emph{kernel} of a
$K$-linear map in the universal algebraic sense, \ie, $\mathrm{ker}(f)=\{(u,v)
\mid f(u)=f(v)\}$ and its universal properties. We show that in the special case
that $K$ is a field, our results are identical to those of \cite{boreale2009weighted}
and \cite{bonchi2012coalgebraic}.

Our coalgebraic characterization of linear bisimulation proceeds as follows. We
define a $K$-linear relation on a semimodule $V$ over the semiring $K$ as the
kernel of a $K$-linear map $f:V \ra W$ with domain $V$. Clearly, this is an
equivalence relation which is a congruence. In fact any relation $R$ on $V$ can
be turned into a $K$-linear relation by considering the smallest congruence
$R^\ell$ containing $R$ and taking $f$ to the be canonical map $V \ra V/R^\ell$
that sends each element to its congruence class. A $K$-linear bisimulation on a
weighted automaton is then defined as a $K$-linear relation which preserves the
output map of the weighted automaton and is invariant under the transition map
as in \cite{bonchi2012coalgebraic}. We show that this approach leads to all the
results of \cite{bonchi2012coalgebraic} regarding the functor $\cL$ and weighted
language equivalence although the proofs are new as we cannot assume the field
properties of $K$, in particular the existence of the additive inverse.
Table~\ref{tab:comp} is a summary of the correspondence between the key concepts
in \cite{boreale2009weighted,bonchi2012coalgebraic} and the present work.

\begin{table}[t]
\centering
\def\arraystretch{1.2}% add some space around the lines
\begin{tabular}{ | m{4cm} | m{5cm}| m{4.5cm} | } 
\hline
% \multicolumn{1}{|r|}{Item3}   & X3 
{\bf Key Component} & {\bf Boreale \cite{boreale2009weighted} and Bonchi \etal\ \cite{bonchi2012coalgebraic} } &
\multicolumn{1}{|c|}{\bf Our work} \\ 
\hline
Domain of weights & Field $K$ & Semiring $K$ \\
\hline
State space & Vector Space $V$ over $K$ & Semimodule $V$ over $K$ \\
\hline
Foundation for $K$-linear relations & Subspace $U$ of $V$ & Linear map $f$ from $V$\\
\hline
$K$-Linear Relation $R$ & $uRv$ iff $u - v \in U$ & $uRv$ iff 
$R= \mathrm{ker}(f)$ \\
\hline
Linear extension $R^\ell$ of relation $R$ & $uR^\ell v$ iff $u-v \in \mathrm{span}(\mathrm{ker}(R))$
& $R^\ell =$ smallest congruence containing $R$ \\
\hline
\end{tabular}
\caption{Correspondence between previous work and ours}
\label{tab:comp}
\end{table}

We briefly summarize our work highlighting the contributions. After introducing
the basic definitions and notation we recall the notion of a $K$-weighted
automaton (Section~\ref{subsec:kwa}) and $K$-weighted bisimulation
(Section~\ref{subsec:wbisim}). We use the definition of Bonchi \etal\
\cite{bonchi2012coalgebraic} but assume that the underlying set of weights is a
semiring $K$. We use the terms $K$-weighted automaton and $K$-weighted
bisimulation to distinguish our setting from that of
\cite{bonchi2012coalgebraic}. The fact that $K$-weighted automata are
$\cW$-coalgebras for an endofunctor $\cW$ on $\mathsf{Set}$ follows, and so does
the correspondence between $K$-weighted bisimulations and kernels of
$\cW$-homomorphims (Section~\ref{subsec:coal-w}). This result uses our definition of
kernels, so there is some novelty here. The proof that $K$-weighted bisimilarity $\sim_w$
coincides with $\cW$-behavioural equivalence $\approx_\cW$ is the one
in \cite{bonchi2012coalgebraic} presented in a different way.

The main focus of our work is on the coalgebras of the functor $\cL$ on
$\mathsf{SMod}$, the category of semimodules over the semiring $K$. We start by
defining the functor $\cL$ and proceed to define the notions of $K$-linear
automata as coalgebras for the functor $\cL$. We define their behaviour in terms
of weighted languages and present the final $\cL$-coalgebra
(Section~\ref{subsec:klwa}). This part is more or less similar to the treatment
in \cite{bonchi2012coalgebraic} as it does not rely on any vector space property
not enjoyed by a semimodule. The point of departure in our work from that of
\cite{bonchi2012coalgebraic} is the definition of $K$-linear relations and
$K$-linear bisimulations (Section~\ref{subsec:klb}) based on the kernel (in the
universal algebraic sense) of $K$-linear maps as mentioned above. We prove the
correspondence between $K$-linear bisimulations and kernels of
$\cL$-homomorphisms and establish the coincidence of the behavioural
equivalence $\approx_\cL$ and weighted language equivalence $\sim_l$
(Section~\ref{subsec:klb}). Although these results mirror those of
\cite{bonchi2012coalgebraic}, the proofs, other than the one for the coincidence
of $\approx_\cL$ and $\sim_l$, are new and more general -- they truly extend the
results from the vector space setting to the semimodule setting in a non-trivial
way. This is where the main contribution of this paper lies.

One desideratum is the existence of a partition refinement algorithm for
computing the weighted language equivalence $\sim_l$ for finitely
generated semimodules. Unfortunately, the results from Bonchi \etal\
\cite{bonchi2012coalgebraic} do not carry over to our semimodule setting. First,
even finitely generated semimodules do not have the descending chain property:
they can have an infinite descending chain of submodules. Also, 
weighted language equivalence is known to be undecidable for finite-state weighted
automata over the tropical semiring ~\cite{krob1994equality,almagor2020s}.
Instead, we offer an abstract procedure via the final sequence whose limit
exists in $\mathsf{SMod}$. The limit of the final sequence is shown to be
isomorphic to the final coalgebra for the functor $\cL$, essentially following
the reasoning in \cite{bonchi2012coalgebraic}. It is well-known that if the
final sequence stabilizes at an object that object is isomorphic to the final
coalgebra. For any $K$-linear weighted automaton there is a cone to the final
sequence such that the kernel of the arrows in the cone constitute a sequence of
$K$-linear relations. We show that these $K$-linear relations converge to
$\cL$-bisimilarity, \ie, language equivalence, in a finite number of steps in
case the state space of the automaton is a finitely generated \emph{Artinian}
semimodule (those satisfying the descending chain property)
as well as when a weaker condition holds (Section~\ref{subsec:lpr}).
Unfortunately, this does not give us an algorithm for computing the bisimilarity
relation in general, as this involves solving $K$-linear equations. Although this is
possible for specific semirings, such as $\R$, $\Z$ and $\N$, the problem is
known to be undecidable for certain semirings~\cite{narendran1996solving}. We
conclude the paper by comparing our work with the only existing coalgebraic
formulation of partition refinement for weighted automata over semirings, that of
K{\"o}nig and K{\"u}pper~\cite{konig2018generalized}
(Section~\ref{subsec:konig}).

\section{$K$-Weighted Automata and $K$-Weighted Bisimulation}
\label{sec:kwa}

This section is a mild generalization of the coalgebraic characterization of
weighted automata and weighted bisimulation in Bonchi \etal\
\cite{bonchi2012coalgebraic} from the setting of vector spaces to semimodules.
We start by fixing the notation and recalling the basic notions of semimodules
and coalgebras. Then we show how weighted automata over a semiring $K$ can be
seen as a coalgebra over a functor $\cW: \mathsf{Set} \ra \mathsf{Set}$ and
characterize weighted bisimilarity $\sim_w$ as the behavioural equivalence
$\approx_\cW$ for the functor $\cW$. The results and proofs are essentially
identical to those in \cite{bonchi2012coalgebraic}, as they do not use the
additional properties satisfied by vector spaces. We include them for the sake
of completeness. The only exception to this is the correspondence between
$K$-weighted bisimulation and kernels of $\cW$-homomorphisms in
Section~\ref{subsec:coal-w}, which uses our definition of kernel and is therefore
new.

\subsection{Notation and Preliminaries}
\label{subsec:prelims}

We denote sets by capital letters $X,Y,Z,\ldots$ and maps (\ie, functions) by
small letters $f,g,h,\ldots$. We denote the identity map on a set $X$ by $id_X$.
Given two maps $f: X \ra Y$ and $g: Y \ra Z$, their composition is denoted $g
\circ f : X \ra Z$. The product of two sets is denoted $X \times Y$ with the
projections $\pi_1 : X \times Y \ra X$ and $\pi_2: X \times Y \ra Y$. The
product of two maps $f_1: X_1 \ra Y_1$ and $f_2: X_2 \ra Y_2$ is $f_1 \times f_2
: X_1 \times X_2 \ra Y_1 \times Y_2$ defined by $(f_1 \times f_2)(x_1,x_2) =
(f_1(x_1),f_2(x_2))$. Given maps $f: X \ra Y$ and $g: X \ra Z$, $\langle f,g
\rangle: X \ra Y \times Z$ is the pairing map defined by $\langle f,g \rangle(x)
= (f(x),g(x))$. We use $\N$ for the set of natural numbers, $\Z$ for the
integers, $\R$ for the set of reals and $\R_+$ for the set of non-negative
reals.

The disjoint union of sets $X_1$ and $X_2$ is $X_1 + X_2$
with the injections $\iota_1 : X_1 \ra X_1 + X_2$ and $\iota_2: X_2 \ra
X_1 + X_2$. The sum of two maps $f_1: X_1 \ra Y_1$ and $f_2: X_2 \ra Y_2$
is $f_1 + f_2 : X_1 + X_2 \ra Y_1 + Y_2$ defined by $(f_1+f_2)(\iota_i(z)) =
\iota_i(f_i(z)))$ for $i=1,2$. We denote the set of maps from $X$ to $Y$ by
$Y^X$. For a map $f: X_1 \ra X_2$, the map $f^Y: X_1^Y \ra X_2^Y$ is defined by
$f^Y(g)=f\circ g$. This defines a functor $(\_)^Y : \mathsf{Set} \ra \mathsf{Set}$,
called the \emph{exponential functor}.
The set of all finite subsets of $X$ is denoted by
$\cP_\omega(X)$. For a finite set of letters $A$, $A^\ast$ denotes the set of
all finite words over $A$. We denote by $\epsilon$ the empty word, and by
$w_1w_2$ the concatenation of words $w_1,w_2 \in A^\ast$. The length of
a word $w$ is denoted by $|w|$.

If $R$ is an equivalence relation on a set $X$ we denote the set of equivalence
classes of $R$ by $X/R$ and the equivalence class of an element $x \in X$ by
$[x]_R$. The subscript is often dropped if the relation $R$ is clear from the
context. For a map $f : X \ra Y$, the \emph{kernel} of $f$ is the
equivalence relation $\mathrm{ker}(f)= \{(x_1,x_2) \mid f(x_1) = f(x_2)\}$.
Further, $f$ has a \emph{unique} (up to isomorphism) factorization through
$X/R$, $f = \mu_f \circ \varepsilon_f$ into a surjection $\varepsilon_f: X \ra
X/\mathrm{ker}(f)$ followed by an injection $\mu_f: X/\mathrm{ker}(f) \ra Y$
defined by $\varepsilon_f(x) = [x]$ and $\mu_f([x])=f(x)$.

\subsection{Semirings and Semimodules}
\label{subsec:semimod}

Semirings and semimodules generalize the notions of fields and vector spaces,
respectively. A \emph{semiring} $(K,+,.,0,1)$ consists of a commutative monoid
$(K,+,0)$ and a monoid $(K,.,1)$ such that the product distributes over the sum
on both sides and $k.0=0.k=0$ for all $k \in K$. We will refer to the semiring
$K$ when the operations are understood. A \emph{semiring module}, or simply a
\emph{semimodule} $V$ over a semiring $K$ is a commutative monoid $(V,+,0)$
together with an action $. : K \times V \ra V$ such that for all
$k,k_1,k_2 \in K$ and $v,v_1,v_2 \in V$:
\begin{align}
&(k_1 + k_2).v = k_1.v + k_2.v \myquad[3] (k_1.k_2).v = k_1.(k_2.v) \\
&k.(v_1+v_2) = k.v_1 + k.v_2 \myquad[3] 1.v = v\\
&0.v = 0
\end{align}
A $K$-linear map between two semimodules $V$ and $W$ (over the semiring $K$) is a
map $f: V \ra W$ satisfying $f(v_1+v_2)=f(v_1) + f(v_2)$ and $f(k.v)=k.f(v)$ for
all $v,v_1,v_2 \in V$ and $k \in K$.

An equivalence relation $R$ on a semimodule $V$  is called a \emph{congruence}
if all the semimodule operations are compatible with $R$, \ie,
for all $k \in K$ and $u_1,u_2,v_1,v_2 \in V$, $u_1Rv_1$ and $u_2Rv_2$ imply
$k.u_1Rk.u_2$ and $(u_1+u_2)R(v_1+v_2)$. It is easy to verify that
the quotient $V/R$ has a semimodule structure given by the operations
$k.[u] = [ku]$ and $[u]+[v] = [u+v]$. The fact that these are well-defined, \ie,
independent of the choice of representative of an equivalence class, is a consequence
of $R$ being a congruence.

For a $K$-linear map $f : V \ra W$, the \emph{kernel} of $f$ is the equivalence
relation $\mathrm{ker}(f)= \{(v_1,v_2) \mid f(v_1) = f(v_2)\}$. It can be shown
that $\mathrm{ker}(f)$ is a congruence. It follows that $V/\mathrm{ker}(f)$ is a
semimodule. Further, every $K$-linear map $f$ has a \emph{unique} (up to
isomorphism) factorization into two $K$-linear maps, a surjection $\varepsilon_f: V \ra
V/\mathrm{ker}(f)$ followed by an injection $\mu_f: V/\mathrm{ker}(f) \ra W$ defined
by $\varepsilon_f(u) = [u]$ and $\mu_f([u])=f(u)$. In the following, when we refer to the kernel
of a $K$-linear map $f$, we mean the kernel of $f$ as defined above and
\emph{not} the set $\{u \mid f(u)=0\}$. This usage is common in universal
algebra~\cite{sankappanavar1981course}.

Let $V$ be a $K$-semimodule. A nonempty subset $U$ of $V$ is called a
\emph{subsemimodule} of $V$ if $U$ is closed under addition and scalar multiplication.
The intersection $\bigcap_{i \in I} U_i$ of any family $\{U_i\}_{i \in I}$ of
subsemimodules of $V$ is clearly a subsemimodule of $V$. If $U$ is any nonempty
subset of $V$ then the intersection of all subsemimodules of $V$ containing $U$
is called the subsemimodule \emph{generated by} $U$, and is denoted by
$\mathrm{span}(U)$. It is easy to check that $\mathrm{span}(U)$ is the set of
finite linear combinations of elements of $U$
\[ \mathrm{span}(U) = \{\sum_{i=1}^n k_i u_i \mid n \in \N, u_i \in U \mbox{ for }
1 \leq i \leq n \}. \]
If $V=\mathrm{span}(U)$ then $U$ is called a \emph{generating set for} $V$. If
$V$ has a finite generating set it is called \emph{finitely generated}. One can
verify that given a $K$-linear map $f:V \ra W$, $\mathrm{ker}(f)$ is a
subsemimodule of $V \times V$.

Semimodules over a semiring $K$ and $K$-linear maps form the category
$\mathsf{SMod}$. $\mathsf{SMod}$ has products $V \times W$, and the set of all
maps $V^A$ from a set $A$ to a semimodule $V$ has a natural semimodule structure
defined pointwise: for $f,g \in V^A$, $f+g \in V^A$ is defined by $(f+g)(a)=
f(a)+g(a)$ and $(k.f)(a)=k.f(a)$. For a set $X$, the set of all maps $f: X
\ra K$ with \emph{finite support}, \ie, the set $\{x \mid f(x) \neq 0\}$ is
finite, is denoted $K(X)$. Its elements are conveniently represented as formal
sums $\sum_{x \in X}k_x.x$ by writing $k_x = f(x)$. In other words,
$K(X)=\mathrm{span}(X)$, where we identify an element $x \in X$ with the
map $\eta_X(x)=\delta_x:X \ra K$ where $\delta_x$ is the Kronecker delta
that maps $x$ to $1$ and everything else to $0$. Note that only a finite number
of $k_x$ are non-zero in the formal sum. $K(X)$ is called the \emph{free
semimodule generated by} $X$ over $K$ and satisfies the following universal
property. Given any map $f : X \ra V$ from a set $X$ to a semimodule $V$ there
exists a unique $K$-linear map $f^\sharp: K(X) \ra V$ which extends $f$. The map
$f^\sharp$ is just the \emph{linear extension} of $f$ \ie, $f^\sharp(\sum_{x \in
X}k_x.x) =\sum_{x \in X}k_x.f(x)$. This is shown in the commuting diagram below
where $\eta_X:X\ra K(X)$ is the inclusion map.
\begin{equation}
\label{diag:free}
\begin{tikzcd}
X \arrow[rd, "f"] \arrow[r, hook, "\eta_X"] & K(X) \arrow[d, dashed, "f^\sharp" near start] \\
& V
\end{tikzcd}
\end{equation}

\subsection{Coalgebras}
\label{subsec:coal}

Given an endofunctor $\cF:\cC \ra \cC$ on a category $\cC$, an $\cF$-coalgebra is a
$\cC$-object $X$ together with a $\cC$-arrow $f:X \ra \cF X$. In many categories
the pair $(X,f)$ represents a transition system such as a deterministic,
nondeterministic or probabilistic automaton~\cite{rutten2000universal}. A
\emph{morphism of $\cF$-coalgebras}, or an \emph{$\cF$-homomorphism}, between
coalgebras $(X,f)$ and $(Y,g)$ is a $\cC$-arrow $h: X \ra Y$ such that the
following diagram commutes.
\begin{equation}
\begin{tikzcd}
X \arrow[r, "h"] \arrow[d, "f"'] & Y \arrow[d, "g"] \\
\cF X \arrow[r, "Fh"] & \cF Y
\end{tikzcd}
\end{equation}
An $\cF$-coalgebra $(Y,g)$ is called \emph{final} is there is a \emph{unique}
$\cF$-homomorphism $\llbracket \_\rrbracket^F_X$ from any $\cF$-coalgebra
$(X,f)$ to $(Y,g)$. The final coalgebra represents the universe of all possible
$\cF$-behaviours and the arrow $\llbracket \_\rrbracket^F_X$ maps every element
(or \emph{state}) of a coalgebra $X$ to its
behaviour~\cite{rutten2000universal}. Two states $x_1,x_2 \in X$ are said to be
\emph{$\cF$-behaviourally equivalent}, denoted $\approx_\cF$, iff $\llbracket
x_1\rrbracket^\cF_X = \llbracket x_2\rrbracket^\cF_X$.

\subsection{$K$-Weighted Automata}
\label{subsec:kwa}

A weighted automaton~\cite{droste2009handbook} is a generalization of a
nondeterministic finite automaton where each transition and each state is
assigned a weight in a semiring $K$. We follow the definition in
\cite{bonchi2012coalgebraic}. Formally, for a semiring $K$, a \emph{$K$-weighted
automaton} ($K$-WA in short) with input alphabet $A$ is a pair $(X, \langle o,t
\rangle)$ where $X$ is a set of states, $o: X \ra K$ is a output map and $t:
X \ra (K^X)^A$ is the transition map. The state $x$ can make a transition to
state $y$ on input $a \in A$ with weight $k \in K$ iff $t(x)(a)(y)=k$. A weight
of zero means there is no transition. Note that the set of states $X$ may be
infinite in general. We often use $X$ to refer to the $K$-WA $(X, \langle o,t
\rangle)$ when the output and transition maps are clear from the context.

\subsection{$K$-Weighted Bisimulation}
\label{subsec:wbisim}

The notion of weighted bisimulation~\cite{buchholz2008bisimulation} generalizes
the well-known notion from ordinary transition systems~\cite{milner1989communication} to
finite-state weighted automata. We follow the definition in
\cite{bonchi2012coalgebraic} which applies to infinite state spaces, but with
\emph{finite branching}: for all $x \in X, a \in A, t(x)(a)(y) \neq 0$ for
only finitely many $y$. In the following we assume the finite branching condition for
weighted automata without stating it explicitly.

\begin{defi} 
Let $M=(X,\langle o,t \rangle)$ be a $K$-weighted automaton. An equivalence
relation $R$ on $X$ is a \emph{$K$-weighted bisimulation} on $M$ if the following
two conditions hold for all $x_1,x_2 \in X$:
\be
\item $o(x_1)=o(x_2)$, and
\item for all $a \in A$ and $y \in X$,
$\sum_{y' \in [y]_R} t(x_1)(a)(y') =  \sum_{y' \in [y]_R} t(x_2)(a)(y')$.
\ee
\end{defi}

The largest $K$-weighted bisimulation relation on $M$ is called \emph{$K$-weighted bisimilarity},
and is denoted by $\sim_w$. It exists because an arbitrary union of $K$-weighted bisimulation
relations is a $K$-weighted bisimulation.

%% PB: Put example of $K$-weighted bisimulation here.

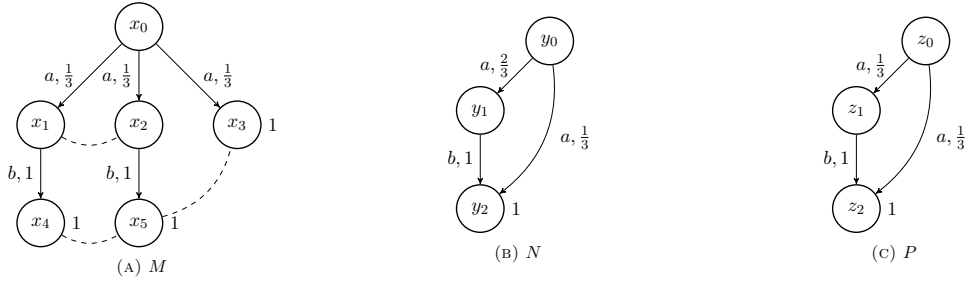
\begin{figure}[t]
\caption{Bisimulation quotient of weighted automata over different semirings}
\vspace{0.5cm}
\label{fig:bisim}
\begin{adjustbox}{width=15cm}
% first automaton
\begin{subfigure}{.5\textwidth}
\centering
\begin{tikzpicture}
\node[state] (q0) {$x_0$}; 
\node[state,below= of q0] (q2) {$x_2$}; 
\node[state,left= of q2] (q1) {$x_1$}; 
\node[state,right = of q2, label=right:$1$] (q3) {$x_3$}; 
\node[state,below= of q1, label=right:$1$] (q4) {$x_4$}; 
\node[state,below= of q2, label=right:$1$] (q5) {$x_5$}; 
\draw (q0) edge[left] node [xshift=-5]{$a,\frac{1}{3}$}(q1);
\draw (q0) edge[left] node{$a,\frac{1}{3}$} (q2);
\draw (q0) edge[right] node [xshift=5]{$a,\frac{1}{3}$} (q3);
\draw (q1) edge[left] node{$b,1$} (q4);
\draw (q2) edge[left] node{$b,1$} (q5);
\draw (q1) edge[-,dashed,bend right] (q2);
\draw (q4) edge[-,dashed,bend right] (q5);
\draw (q5) edge[-,dashed,bend right] (q3);
\end{tikzpicture}
\caption{$M$}
%\label{fig:M}
\end{subfigure}%
%\hspace{1cm}
% second automaton
\begin{subfigure}{.5\textwidth}
\centering
\begin{tikzpicture}
\node[state] (p0) {$y_0$}; 
\node[state,below left= of p0] (p1) {$y_1$}; 
%\node[state,below right= of p0, label=right:$1$] (p2) {$y_2$}; 
\node[state,below= of p1, label=right:$1$] (p2) {$y_2$}; 
\draw 
	(p0) edge[left,pos=0.2] node [xshift=-5] {$a,\frac{2}{3}$} (p1)
 	(p0) edge[bend left, right] node [xshift=5]{$a,\frac{1}{3}$} (p2)
 	(p1) edge[left] node{$b,1$} (p2);
\end{tikzpicture}
\caption{$N$}
\end{subfigure}%
%\hspace{1cm}
% third automaton
\begin{subfigure}{.5\textwidth}
\centering
\begin{tikzpicture}
\node[state] (r0) {$z_0$}; 
\node[state,below left= of r0] (r1) {$z_1$}; 
%\node[state,below right= of r0, label=right:$1$] (r2) {$z_2$}; 
\node[state,below= of r1, label=right:$1$] (r2) {$z_2$}; 
\draw 
	(r0) edge[left,pos=0.2] node [xshift=-5]{$a,\frac{1}{3}$} (r1)
 	(r0) edge[bend left, right] node [xshift=5]{$a,\frac{1}{3}$} (r2)
 	(r1) edge[left] node{$b,1$} (r2);
\end{tikzpicture}
\caption{$P$}
\end{subfigure}%
%\end{center}      
\end{adjustbox}  
\end{figure}

\begin{exa}
This example illustrates bisimulation between weighted automata for different
semirings, and is adapted from \cite{buchholz2008bisimulation}.
Figure~\ref{fig:bisim}(A) shows a $K$-weighted automaton $M=(X_\cA, \langle
o_\cA,t_\cA \rangle)$ over the alphabet $A = \{a,b\}$ and the semiring
$K=(\R_+,+,\cdot,0,1)$. It has the set of states $X_M =
\{x_0,x_1,x_2,x_3,x_4,x_5\}$, the output map (shown in the figure as labels of
the states when they are non-zero) $o_\cA$ given by $\{x_0 \mapsto 0, x_1
\mapsto 0, x_2 \mapsto 0, x_3 \mapsto 1, x_4 \mapsto 1, x_5 \mapsto 1\}$ and the
transition map $t$ shown in the figure by the edges with their labels. For
example, $t(x_0)(a)(x_1) = \frac{1}{3}$. All missing transitions have weight
zero. Let $R$ be the smallest equivalence relation on $X_M$ containing the
pairs $\{(x_1,x_2),(x_3,x_5),(x_4,x_5)\}$. $R$ is shown in the figure as the dashed lines,
which join states in the same equivalence class. It is easily checked that $R$ is a
bisimulation relation on $M$, and is in fact the largest such. The automaton
$N$ in Figure~\ref{fig:bisim}(B) is the one obtained from $M$ by quotienting
by $R$, where the quotient operation is defined in Section~\ref{subsec:coal-w}.
Here the set of states is $\{y_0,y_1,y_2\}$ and the output map is
given by $\{ y_0 \mapsto 0, y_1 \mapsto 0, y_2 \mapsto 1\}$.

Now consider a different semiring $K'=([0,1],\max,\cdot,0,1)$ and consider $M$,
the automaton on the left, as a $K'$-weighted automaton. The relation $R$ above
is again the largest bisimulation on $M$, but the quotient automaton $P$ has
different weights on transitions, and is shown in Figure~\ref{fig:bisim}(C).
Here the set of states is $\{z_0,z_1,z_2\}$ and the output map is given by $\{
z_0 \mapsto 0, z_1 \mapsto 0, z_2 \mapsto 1\}$. Notice the difference in the
weights on the edges $(y_0,y_1)$ and $(z_0,z_1)$. The former is obtained by
addition whereas the latter by the $\max$ operation on the pair
$(\frac{1}{3},\frac{1}{3})$.
\end{exa}

\subsection{Coalgebraic Model for $K$-Weighted Automata and $K$-Weighted Bisimulation}
\label{subsec:coal-w}

Following \cite{bonchi2012coalgebraic}, we now exhibit a functor $\cW:
\mathsf{Set} \ra \mathsf{Set}$ such that a $\cW$-coalgebra is just a
$K$-weighted automaton and $\approx_\cW$ is exactly $K$-weighted bisimilarity.

\begin{defi}
For a semiring $K$ the \emph{valuation functor} $K(\_): \mathsf{Set} \ra \mathsf{Set}$
is defined by the mappings $X \mapsto K(X)$ on sets $X$ and
$X\stackrel{h}{\ra} Y \mapsto K(X) \stackrel{K(h)}{\ra} K(Y)$ on maps,
where $K(h)$ sends $\sum_{x \in X} k_xx \in K(X)$ to
$\sum_{y \in Y} k_yy \in K(Y)$ with $k_y = \sum_{x \in h^{-1}(y)}k_x$.
\end{defi}

Recall that for a given set $C$, the functor $C \times \_: \mathsf{Set} \ra
\mathsf{Set}$ sends a set $X$ to $C\times X$ and a map $f:X\ra Y$ to the map
$id_C\times f$. The functor $\cW: \mathsf{Set} \ra \mathsf{Set}$ is
defined by $\cW = K \times (K(\_))^A$ where $(\_)^A$ is the exponential
functor defined earlier. Thus a $\cW$-coalgebra $f: X \ra \cW X$ on
a set $X$ constitutes a pair of maps $\langle o,t \rangle$ where $o: X \ra
K$ and $t: X \ra K(X)^A$. In other words, a $\cW$-coalgebra is identical to a $K$-weighted
automaton $(X,\langle o,t \rangle)$ and vice versa under the assumption of finite branching,
since $K(X)$ is the set of maps $G : X \ra K$ with finite support, which means $t$
satisfies the finite branching property.

It is shown in Bonchi \etal\ \cite{bonchi2012coalgebraic} that the functor $\cW$, being
\emph{bounded}, has a final coalgebra $(\Omega,\omega)$.
Moreover, the behavioural equivalence $\approx_\cW$ coincides with $K$-weighted
bisimilarity $\sim_w$. It is also shown that $K$-weighted bisimilarity is
strictly included in weighted language inclusion. For completeness we
recall the proof by Bonchi \etal\ of the coincidence of the two relations
$\approx_\cW$ and $\sim_w$.

Recall that a map $h: X \ra Y$ is a $\cW$-homomorphism between weighted
automata, \ie, $\cW$-coalgebras, $(X,\langle o_X,t_X \rangle)$ and $(Y,\langle
o_Y,t_Y \rangle)$ when the following diagram commutes.
\begin{equation}
\begin{tikzcd}
X \arrow[r, "h"] \arrow[d, "{\langle o_X, t_X \rangle}"'] & 
	Y \arrow[d, "{\langle o_Y, t_Y \rangle}"]  \\
K \times K(X)^A \arrow[r,"{id_K \times K(h)^A}"'] & K \times K(Y)^A
\end{tikzcd}
\end{equation}
In words, for all $x \in X$, $y \in Y$, $a \in A$ \[o_X(x)=o_Y(h(x)) \myquad[2]
\mbox{and} \myquad[2] \sum_{x' \in h^{-1}(y)} t_X(x)(a)(x') = t_Y(h(x)(a)(y). \]
For any $\cW$-homomorphism $h: (X,\langle o_X,t_X \rangle) \ra (Y,\langle
o_Y,t_Y \rangle)$, the equivalence relation $\mathrm{ker}(h)$ is a weighted
bisimulation since $h(x_1)=h(x_2)$ implies 
\[ o_X(x_1)=o_Y(h(x_1))=o_Y(h(x_2))=o_X(x_2) \]
and for all $a \in A$, for all $y \in Y$ 
\[\sum_{x'' \in h^{-1}(y)}t_X(x_1)(a)(x'')=t_Y(h(x_1))(a)(y)=t_Y(h(x_2))(a)(y)=
\sum_{x'' \in h^{-1}(y)}t_X(x_2)(a)(x'') \]
which in turn implies that for all $x' \in X$ 
\[ \sum_{(x',x'') \in \mathrm{ker}(h)}t_X(x_1)(a)(x'') = \sum_{(x',x'') \in
\mathrm{ker}(h)}t_X(x_2)(a)(x''). \]
Conversely, every $K$-weighted bisimulation $R$ on $(X,\langle o_X,t_X \rangle)$
induces a coalgebra structure $(X/R,\langle o_{X/R},t_{X/R} \rangle)$ on the
quotient set $X/R$ where $o_{X/R}: X/R \ra K$ and $t_{X/R}: X/R \ra (X/R)^A$ are
defined by 
\[ o_{X/R}[x]= o_X(x) \myquad[2] \mbox{and} \myquad[2]
t_{X/R}[x_1](a)([x_2]) = \sum_{x' \in [x_2]}t_X(x_1)(a)(x'). \] 
As $R$ is a $K$-weighted bisimulation, both $o_{X/R}: X/R \ra K$ and $t_{X/R}:
X/R \ra (X/R)^A$ are well-defined, \ie, independent of the choice of representative
of an equivalence class. Now, the map $\varepsilon_R: X \ra X/R$ which
sends $x$ to its equivalence class $[x]_R$ is a $\cW$-homomorphism. Therefore we
have the following commuting diagram, where the dashed arrows constitute the
unique $\cW$-homomorphisms to the final coalgebra $(\Omega,\omega)$.
\begin{equation}
\label{diag:final}
\begin{tikzcd}[row sep = large, row sep = large]
X \arrow[r, "\varepsilon_R"] 
\arrow[rr, bend left, dashed, "{\llbracket \_\rrbracket^\cW_X}"]
\arrow[d, "{\langle o_X, t_X \rangle}"'] 
&
X/R 
\arrow[d, "{\langle o_{X/R}, t_{X/R} \rangle}"] 
\arrow[r, dashed, "{\llbracket \_\rrbracket^{\cW}_{X/R}}"]
&
\Omega 
\arrow[d, "\omega"] 
\\
\cW(X) \arrow[r,"{\cW(\varepsilon_R)}"'] 
\arrow[rr, bend right, dashed, "{\cW(\llbracket \_\rrbracket^\cW_{X})}"]
& \cW(X/R) \arrow[r, dashed, "{\cW(\llbracket \_\rrbracket^{\cW}_{X/R})}"]
& \cW(\Omega)
\end{tikzcd}
\end{equation}

\begin{thm}
\label{them:equiv}
Let $(X \langle o,t \rangle)$ be a weighted automaton. Then for $x_1,x_2 \in X$,
$x_1 \sim_w x_2$ iff $x_1 \approx_\cW x_2$.
\end{thm}

\begin{proof}
The equality of the two relations $\sim_w$ and $\approx_\cW$ follows by
diagram chasing. By definition, $\approx_\cW = \mathrm{ker}(\llbracket
\_\rrbracket^\cW_X)$ and this is is a $K$-weighted bisimulation since $\llbracket
\_\rrbracket^\cW_X$ is a $\cW$-homomorphism as witnessed by the curved and
dashed arrows in the diagram, which implies $\approx_\cW \subseteq \sim_w$. In the other
direction,
\begin{flalign*}
&x \sim_w x_2 \\
&\Ra (x_1,x_2) \in R \mbox{ for a $K$-weighted bisimulation } R \\
&\Ra \varepsilon_R(x_1) = \varepsilon_R(x_2), \mbox{ since } [x_1]_R = [x_2]_R \\
&\Ra \llbracket \varepsilon_R(x_1) \rrbracket^{\cW}_{X/R}
= \llbracket \varepsilon_R(x_2)\rrbracket^{\cW}_{X/R}\\
&\Ra \llbracket x_1 \rrbracket^\cW_X = \llbracket x_2 \rrbracket^\cW_X
\mbox{ from the diagram above} \\
&\Ra x_1 \approx_\cW x_2, \mbox{ by definition of } \approx_\cW. \tag*{\qedhere}
\end{flalign*}
\end{proof}

We now define $K$-weighted language equivalence $\sim_l$ for $K$-weighted automata
and show that $K$-weighted bisimilarity is a refinement of $K$-weighted 
language equivalence. The proof is from \cite{bonchi2012coalgebraic}.

A $K$-\emph{weighted language} over an alphabet $A$ and semiring $K$ is a map
$\sigma: A^\ast \ra K$ that assigns to each word $w \in A^\ast$ a weight in $K$.
For a $K$-WA $(X, \langle o,t \rangle)$ the map $\sigma_l: X \ra K^{A^\ast}$
assigns to each state $x \in X$ the $K$-weighted language \emph{recognized by}
$x$ and is defined by induction on $w$ as follows.
\[
\sigma_l(x)(w) =
\begin{cases}
o(x) \myquad[10]\,\,\,\, &\mbox{if } w = \epsilon \\
\sum_{x' \in X}(t(x)(a)(x')).\sigma_l(x')(w')  \quad &\mbox{if } w = aw'
\end{cases}
\]
Two states $x_1,x_2 \in X$ are said to be weighted language equivalent, denoted
$x_1 \sim_l x_2$, if $\sigma_l(x_1)(w) = \sigma_l(x_2)(w)$ for all $w \in A^\ast$.

\begin{prop}
\label{prop:refine}
For $K$-weighted automata, $\sim_w \subseteq \sim_l$.
\end{prop}

\begin{proof}
We prove by induction on the length of $w$ that if $R$ is a $K$-weighted bisimulation
on $X$ then for all $x_1,x_2 \in X$ and all $w \in A^\ast$, $(x_1,x_2) \in R$ implies
$\sigma_l(x_1)(w) = \sigma_l(x_2)(w)$. For the base case $w=\epsilon$, we have
$\sigma_l(x_1)(w) = o(x_1)$ and $\sigma_l(x_2)(w) = o(x_2)$ and $o(x_1)=o(x_2)$
since $R$ is a $K$-weighted bisimulation. For the inductive case, if $w = aw'$ then
\begin{flalign*}
\sigma_l(x_1)(w) =& \sum_{x' \in X}(t(x_1)(a)(x')).\sigma_l(x')(w') \\
=& \sum_{[x']_R \in X/R}((\sum_{x'' \in [x']_R} t(x_1)(a)(x'')).\sigma_l(x')(w'))
\mbox{ since by the induction hypothesis}\\
& \quad \mbox{for all } x'' \in [x']_R, \sigma_l(x'')(w')=\sigma_l(x')(w')
\mbox{ and by grouping the states} \\
& \quad x'' \in [x'] \\
=& \sum_{[x']_R \in X/R} ((\sum_{x'' \in [x']_R} t(x_2)(a)(x'')).\sigma_l(x')(w'))
\mbox{ since $(x_1,x_2) \in R$ and $R$ is a $K$-} \\
& \quad \mbox{weighted bisimulation} \\
=&\quad  \sigma_l(x_2)(w) \mbox{ by an argument similar to the first two lines above.} \tag*{\qedhere}
\end{flalign*}
\end{proof}

\begin{figure}[t]
\caption{Weighted language equivalence and weighted bisimulation}
\vspace{0.5cm}
\label{fig:lang-bisim}
\begin{adjustbox}{width=10cm}
% first automaton
%\begin{subfigure}{.5\textwidth}
%\centering
\begin{tikzpicture}
\node[state] (q0) {$x_0$}; 
\node[state,below left= of q0] (q1) {$x_1$}; 
\node[state,below right = of q0] (q2) {$x_2$}; 
\node[state,below= of q1, label=right:$1$] (q3) {$x_3$}; 
\node[state,below= of q2, label=right:$1$] (q4) {$x_4$}; 
\draw (q0) edge[left,pos=0.2] node [xshift=-5]{$a,\frac{1}{2}$} (q1);
\draw (q0) edge[right,pos=0.2] node [xshift=5]{$a,\frac{1}{2}$} (q2);
\draw (q1) edge[left] node{$b,1$} (q3);
\draw (q2) edge[right] node{$c,1$} (q4);
\end{tikzpicture}
%\caption{$M$}
%%\label{fig:M}
%\end{subfigure}%
\hspace{1cm}
% second automaton
%\begin{subfigure}{.5\textwidth}
%\centering
\begin{tikzpicture}
\node[state] (p0) {$y_0$}; 
\node[state,below = of p0] (p1) {$y_1$}; 
\node[state,below left= of p1, label=right:$1$] (p2) {$y_2$}; 
\node[state,below right= of p1, label=right:$1$] (p3) {$y_3$}; 
\draw 
	(p0) edge[left] node{$a,\frac{1}{2}$} (p1)
 	(p1) edge[left,pos=0.3] node{$b,1$} (p2)
 	(p1) edge[right,pos=0.3] node{$c,1$} (p3);
\end{tikzpicture}
%\caption{$N$}
%\end{subfigure}%  
\end{adjustbox}  
\end{figure}
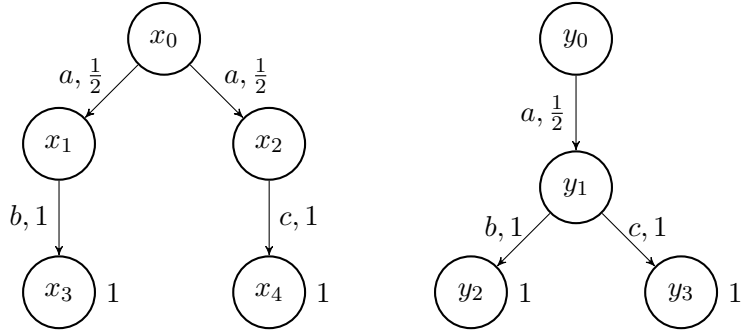

\begin{exa}
Figure~\ref{fig:lang-bisim} is an adaptation of a familiar example from the process
algebra literature that shows that $K$-weighted bisimilarity strictly
refines $K$-weighted language equivalence. Here $K=(\R_+,+,\cdot,0,1)$. The
states $x_0$ and $y_0$ are language equivalent, since $\sigma_l(x_0)(ab)=
\sigma_l(y_0)(ab) = \frac{1}{2}$ and $\sigma_l(x_0)(ac) = \sigma_l(y_0)(ac) =
\frac{1}{2}$ and $\sigma_l(x_0)(w)=\sigma_l(y_0)(w)=0$ for all other words $w$.
But it is easily checked that $x_0$ and $y_0$ are not bisimilar.
\end{exa}

\section{$K$-Linear weighted automata as Coalgebras over Semimodules}
\label{sec:$K$-LWA}

The goal of this section is to show that there is a functor $\cL: \mathsf{SMod} \ra
\mathsf{SMod}$ for which behavioural equivalence $\approx_\cL$ coincides with
weighted language equivalence $\sim_l$ of $K$-linear weighted automata, extending the
results of Bonchi \etal\ \cite{bonchi2012coalgebraic} to the category of
semimodules. Although the functor $\cL$ appears to be the same as in
\cite{bonchi2012coalgebraic} the underlying details are different. The latter are
based on a generalization of the notion of a \emph{linear weighted automaton}
in \cite{boreale2009weighted,bonchi2012coalgebraic} from the setting of vector
spaces to semimodules. We propose definitions of a $K$-linear relation and a
$K$-linear bisimulation that generalize the notions of a linear relation and a
linear bisimulation from \cite{boreale2009weighted}. This is the central part of
the paper where the definitions and proofs do not mirror those in
\cite{boreale2009weighted,bonchi2012coalgebraic}. In particular the notion of
subspace of a vector space is replaced by that of the kernel (in the universal
algebraic sense) of a $K$-linear map. But it is remarkable that all proofs go
through and we obtain a true generalization of the concepts from vector spaces
to semimodules.

\subsection{$K$-Linear Weighted Automata}
\label{subsec:klwa}

The following definition is a generalization of a linear weighted automaton of
\cite{boreale2009weighted} from the setting of vector spaces and linear maps to
that of semimodules over a semiring $K$ and $K$-linear maps. Note that we use
the term ``$K$-linear weighted automaton" to distinguish it from the ``linear
weighted automaton" of \cite{boreale2009weighted,bonchi2012coalgebraic}.

\begin{defi}
A \emph{$K$-linear weighted automaton} ($K$-LWA in short) with input alphabet $A$ over
the semiring $K$ is a coalgebra for the functor $\cL = K \times (\_)^A:
\mathsf{SMod} \ra \mathsf{SMod}$.
\end{defi}

A $K$-LWA can be presented as a pair $(V,\langle o,t \rangle)$ where $V$ is a
semimodule over $K$ whose elements are called \emph{states}, $o: V \ra K$ is a
$K$-linear map assigning an \emph{output weight} to every state and $t:V\ra V^A$
is a $K$-linear \emph{transition map} that, given a current state $v$ and input
$a$, assigns a new state $t(v)(a)$. We write $\trans{v_1}{a}{v_2}$ for
$t(v_1)(a)=v_2$. We often use $V$ to refer to the $K$-LWA $(V,\langle o,t \rangle)$
when the output and transition maps are clear from the context.

The behaviour of $K$-LWA is described by weighted languages. The $K$-\emph{linear
weighted language recognized} by a state $v \in V$ of a $K$-LWA $(V,\langle o,t \rangle)$
is the map $\sigma_l(v): A^\ast \ra K$ defined by induction on words $\cW$ by:
\[
\sigma_l(v)(w) =
\begin{cases}
o(v),\myquad[5]\,\,\, &\mbox{if } w = \epsilon\\
\sigma_l(t(v)(a))(w') \myquad[2] &\mbox{if } w = aw'
\end{cases}
\]
Two states $v_1,v_2 \in V$ are said to be weighted language equivalent, denoted
$v_1 \sim_l v_2$, if $\sigma_l(v_1)(w) = \sigma_l(v_2)(w)$ for all $w \in A^\ast$.
Note that we overload the symbol $\sigma_l$ to denote weighted language
equivalence for both $K$-weighted automata and $K$-linear weighted automata.
The context disambiguates which concept the symbol denotes. Later in this
section we show that $\sigma_l(v) = \llbracket v \rrbracket^\cL_V$,
the image of $v$ under the unique $\cL$-homomorphism from $V$ into the final 
$\cL$-coalgebra.

Given a $K$-WA $(X,\langle o,t \rangle)$ (see Section~\ref{subsec:kwa}), we can
construct a $K$-LWA $(K(X),\langle o^\sharp,t^\sharp \rangle)$, where $K(X)$ is
the free semimodule generated by $X$ and $o^\sharp$ and $t^\sharp$ are linear
extensions of $o$ and $t$. It can be shown that the above $K$-WA $X$ and the
$K$-LWA $K(X)$ have equivalent language behaviour, \ie, the corresponding
states $x$ and $\eta_X(x)$ recognize the same weighted language for all $x \in X$.

%% PB: Put example of $K$-LWA here.

Recall that a $K$-linear map $h: V \ra W$ is an $\cL$-homomorphism between
$K$-LWA $(V,\langle o_V,t_V \rangle)$ and $(W,\langle o_W,t_W \rangle)$
when the following diagram commutes.
\begin{equation}
\begin{tikzcd}
V \arrow[r, "h"] \arrow[d, "{\langle o_V, t_V \rangle}"'] & 
	W \arrow[d, "{\langle o_W, t_W \rangle}"]  \\
K \times V^A \arrow[r,"{id_K \times h^A}"'] & K \times W^A
\end{tikzcd}
\end{equation}
In words, for all $v \in V$, $a \in A$, $o_V(v)=o_W(h(v))$ and $h(t_V(v)(a))=
t_W(hv)(a)$. 

For the special case when the $K$-LWA $V=K(X)$ and $W=K(Y)$ for given $K$-WA $X$
and $Y$ as above, we have the following situation. For a map $h:X \ra Y$, the
map $K(h): K(X) \ra K(Y)$ is the unique linear extension of $\eta_Y \circ h: X
\ra K(Y)$ and is hence $K$-linear. If $h$ is a $\cW$-homomorphism between the $K$-WA
$(X, \langle o_X,t_X \rangle)$ and $(Y, \langle o_Y,t_Y \rangle)$ then $K(h)$ is
an $\cL$-homomorphism between the $K$-LWA $(K(X), \langle o^\sharp_X,t^\sharp_X
\rangle)$ and $(K(Y), \langle o^\sharp_Y,t^\sharp_Y \rangle)$.

Bonchi \etal\ \cite{bonchi2012coalgebraic} showed that the final $\cL$-coalgebra
is defined on the set of all weighted languages $K^{A^\ast}$ as follows.
Consider the structure $(K^{A^\ast},\epsilon,d)$ with the output map $\epsilon$
and the transition map $d$ where $\epsilon:K^{A^\ast} \ra K$, called the
\emph{empty map}, is defined by $\epsilon(\sigma)=\sigma(\epsilon)$ and
$d:K^{A^\ast} \ra (K^{A^\ast})^A$ is defined by $d(\sigma)(a)= \sigma_a$ where
$\sigma_a: A^\ast \ra K$ is the \emph{$a$-derivative} of $\sigma$:
\[ \sigma_a(w)=\sigma(aw).\]
We first show that the map $d$ is $K$-linear. If $\sigma_1$ and $\sigma_2$ are two
weighted languages in $K^{A^\ast}$, $k_1, k_2 \in K$, $a \in A$ and $w \in A^\ast$
then
\begin{flalign*}
d(k_1\sigma_1 + k_2\sigma_2)(a)(w) &= (k_1\sigma_1 + k_2\sigma_2)(aw) \\
&= (k_1\sigma_1)(aw) + (k_2\sigma_2)(aw) \\
&= k_1\sigma_1(aw) + k_2\sigma_2(aw) \\
&= k_1d(\sigma_1)(a)(w) + k_2d(\sigma_2)(a)(w) 
\end{flalign*}
as desired. The proof of $K$-linearity of $\epsilon$ is similar. Hence
$(K^{A^\ast},\epsilon,d)$ is a coalgebra in $\mathsf{SMod}$. We now recall the
proof from \cite{bonchi2012coalgebraic} that it is the final coalgebra in $\mathsf{SMod}$.

\begin{thm}
There exists a unique $\cL$-homomorphism from any coalgebra $(V, \langle o,t \rangle)$
into the coalgebra $(K^{A^\ast},\epsilon,d)$.
\end{thm}

\begin{proof}
It is easy to check that the map $\llbracket \_ \rrbracket^\cL_V=\sigma_l:
V \ra K^{A^\ast}$ which maps every state $v \in V$ to the weighted language
$\sigma_l(v)$ is the only one that makes the following diagram commute in
$\mathsf{Set}$.

\begin{equation}
\begin{tikzcd}
V \arrow[r, "{\llbracket \_ \rrbracket^\cL_V}"] \arrow[d, "{\langle o, t \rangle}"'] & 
	K^{A^\ast} \arrow[d, "{\langle \epsilon,d \rangle}"]  \\
\cL(V)  \arrow[r,"{\cL(\llbracket \_ \rrbracket^\cL_V)}"'] & \cL(K^{A^\ast})
\end{tikzcd}
\end{equation}

To show that $\llbracket \_ \rrbracket^\cL_V$is $K$-linear we prove
that $\llbracket k_1v_1 + k_2v_2 \rrbracket^\cL_V(w)
= \llbracket k_1v_1 \rrbracket^\cL_V(w) + \llbracket k_2v_2 \rrbracket^\cL_V(w)$
for all $w \in A^\ast$ by a routine induction on the length of the word $w$.
\end{proof}

It follows that two states $v_1,v_2 \in V$ are $\cL$-behaviourally equivalent, \ie,
$v_1 \approx_{\cL} v_2$ iff they recognize the same weighted language. 

%%%%%%%%%%%%%%%%%%%%%%%%%%%%%%%%%%%%
%%% PB: New stuff
%%%%%%%%%%%%%%%%%%%%%%%%%%%%%%%%%%%%

\subsection{$K$-Linear Bisimulation}
\label{subsec:klb}

In this section we generalize the definition of Boreale's linear weighted
bisimulation~\cite{boreale2009weighted} from a field to a semiring $K$. We show
that the two notions coincide in the special case when $K$ is a field, for
example $K =\R$, as in \cite{boreale2009weighted}. Starting with this section
almost all the results are our contribution and involve new concepts and proofs.

\begin{defi}
A binary relation $R$ on a $K$-semimodule $V$ is \emph{$K$-linear} if there
exists a $K$-semimodule $W$ and a $K$-linear map $f : V \ra W$ such that
$R=\mathrm{ker}{f} =\{(u,v) \mid f(u)=f(v)\}$. Such a relation is denoted by
$R_f$ for the given $f$. 
\end{defi}

It is immediate that a $K$-linear relation on a $K$-semimodule $V$ is an
equivalence relation which, in addition, is a congruence. Moreover, there is a
canonical way of turning \emph{any} relation $R$ on $K(X)$ into a $K$-linear relation
$R^\ell$ as follows. Let $R^\ell$ be the least congruence relation on $V$
containing $R$. $R^\ell$ is obtained by taking the intersection of all
congruences on $V$ that contain $R$ and is well-defined since the universal
relation is a congruence and the intersection of any family of congruences is a
congruence. The quotient set $V/R^\ell$ has a $K$-semimodule structure given by
$[u]+[v]=[u+v]$ and $k.[u]= [k.u]$, which are well-defined since $R^\ell$ is a
congruence. Let $f=\varepsilon_{R^\ell} : X \ra V/R^\ell$ be the map which sends
any elements $v \in V$ to its equivalence class $[v]_{R^\ell}$. It is easy to
check that $f$ is a $K$-linear map and $R^\ell = \mathrm{ker}(f)$ by
construction. Hence $R^\ell$ is a $K$-linear relation. The following lemma is an
easy consequence of the definitions.

\begin{lem}
For any binary relation $R$ on a $K$-module $V$, $R^\ell$ is the smallest $K$-linear relation
containing $R$.
\end{lem}

Note that for a given $K$-linear relation $R$ there may be two distinct $K$-linear maps
$f,g: V \ra W$ with the same codomain such that $R = \mathrm{ker}(f) =
\mathrm{ker}(g)$. On the other hand, all such maps factor uniquely through the
map $\varepsilon_f: V \ra V/\mathrm{ker}{f}$ that sends $v \in V$ to its
equivalence class $[v]$ in $\mathrm{ker}{f}$. Also, note that for any injective map
$f: V \ra W$, $R_f$ is the identity relation on $V$. For the zero map $0_{V,W}: V
\ra W$ which maps every element in $V$ to $0$, $R_{0_{V,W}}$ is the universal
relation on $V$.

We are now ready to define a $K$-linear bisimulation in analogy with linear bisimulation
in \cite{boreale2009weighted,bonchi2012coalgebraic}.

\begin{defi} Let $(V,\langle o,t\rangle)$ be a $K$-LWA, for a semimodule $K$.
A $K$-linear relation $R$ on $V$ is a $K$-\emph{linear bisimulation} if for all $(v_1,v_2)$ 
the following holds:
\be
\item $o(v_1)=o(v_2)$, and
\item $\forall a \in A, t(v_1)(a)\, R \, t(v_2)(a).$
\ee
\end{defi}

When $V$ is a finite dimensional vector space over the field $K$ the notions of
$K$-linear relations and the linear relations of Boreale coincide. Consider a linear
relation (as defined in \cite{boreale2009weighted}) $R$ over the vector space
$V=K(X)$ with basis $X$ where $K$ is a field. By definition, there exists a
subspace $U$ of the vector space $V$ over $K$ such that $uRv$ iff $u-v \in U$.
It is easy to check that the equivalence relation $R=R_U=\{(u,v) \mid u-v \in
U\}$ is a congruence. Then consider the canonical linear map $f_U: V \ra V/R_U$
to the quotient space which maps an element $w \in V$ to $[w]_R$. Now $w \in U$
iff $w = u-v$ for some $u,v \in V$ with $uRv$. Therefore, $f_U(w)=
f_U(u-v)=f_U(u) - f_U(v) = 0$ since $[u]=[v]$, \ie, $f_U$ sends all elements in
$U$ to $0$. It follows that $u-v \in U$ iff $f_U(u)=f_U(v)$ and thus $R$ is a
$K$-linear relation. Conversely, if $R$ is a $K$-linear relation over $V$ then $R=R_f$ for
some $f : V \ra W$. Let $U = \{w \mid f(w)=0\}$. Clearly $U$ is a subspace of
$V$ and $u -v \in U$ iff $f(u)=f(v)$, \ie, $R$ is a linear relation. In
addition, when $V$ is a vector space over a field $K$ it is routine to verify
that the notions of a linear weighted bisimulation as in
\cite{boreale2009weighted,bonchi2012coalgebraic} and a $K$-linear bisimulation
coincide, as the two definitions are identical. Hence, $K$-linear bisimulation
is a more general notion.

The following characterization of $K$-linear bisimulation is immediate from the definition.

\begin{lem}
\label{lem:kbisim}
Let $(V,\langle o,t\rangle)$ be a $K$-LWA, where $V$ is a
$K$-semimodule and $R$ a $K$-linear relation on $V$. Then $R$ is a $K$-linear bisimulation iff
\be
\item[(1)] $R \subseteq \mathrm{ker}(o)$, and
\item[(2)] $R$ is $t_a$-invariant, \ie, $uRv$ implies $t_auRt_av$, for each $a \in A$.
\ee
\end{lem}

More generally, the kernel of a $\cL$-homomorphism between two $K$-LWA is a
$K$-linear bisimulation and conversely, for each $K$-linear bisimulation $R$ there exists a
$\cL$-homomorphism between two $K$-LWA whose kernel is $R$. This result mirrors
the one in \cite{bonchi2012coalgebraic} for linear weighted bisimulation between linear weighted
automata, albeit with a different notion of kernel.

\begin{prop}
\label{prop:lhomo}
Let $(V,\langle o_V,t_V \rangle)$ and $(W,\langle o_\cW,t_\cW \rangle)$ be two
$K$-LWA and $h: V \ra W$ an $\cL$-homomorphism. Then $\mathrm{ker}(h)$ is an
$K$-linear bisimulation on $(V,\langle o_V,t_V \rangle)$. Conversely, if $R$ is a
$K$-linear bisimulation on $(V,\langle o_V,t_V \rangle)$ then there exists a $K$-LWA
$(W,\langle o_W,t_W \rangle)$ and a $\cL$-homomorphism $h: V \ra W$ such that
$R=\mathrm{ker}(h)$.
\end{prop}

\begin{proof}
Suppose $h: V \ra W$ is an $\cL$-homomorphism between $(V,\langle o_V,t_V
\rangle)$ and $(W,\langle o_W,t_W \rangle)$. We show that $\mathrm{ker}(h)$
satisfies clauses (1) and (2) of Lemma~\ref{lem:kbisim}. Take any $(v,w) \in
\mathrm{ker}(h)$. Since by definition $h(v)=h(w)$, we have $o_W(h(v)) = o_W(h(w))$ and
$t_W(h(v))(a) = t_W(h(w))(a)$ for all $a \in A$. Since $h$ is an
$\cL$-homomorphism, we have (1) $o_V(v) = o_W(h(v)) = o_W(h(w)) = o_V(w)$,
\ie, $\mathrm{ker}(h) \subseteq \mathrm{ker}(o_V)$ and (2) $h(t_V(v)(a)) =
t_W(h(v))(a) = t_W(h(w))(a) = h(t_V(w)(a))$, which means $(t_V(v)(a),
t_V(w)(a) \in \mathrm{ker}(h)$ \ie, $\mathrm{ker}(h)$ is $t_a$-invariant.

In the other direction, let $R$ be $K$-linear bisimulation on $(V,\langle o_V,t_V
\rangle)$, where $R=R_f$ for the map $f:V \ra W$. Let $f = \mu_f \circ
\varepsilon_f$ be the unique factorization of $f$ through 
$\varepsilon_f : V \ra V/\mathrm{ker}(f)$ that sends each $v \in V$ to its
equivalence class in $\mathrm{ker}(f)$. As we observed earlier,
$W=V/\mathrm{ker}(f)$ can be equipped with a $K$-LWA structure $(W,\langle o_W,t_W
\rangle)$ as follows. The $K$-linear map $o_W: W \ra K$ is defined as
$o_W([v])=o_V(v)$. The $K$-linear map $t_W:W \ra W^A$ defined as $t_W([v])(a) =
t_V(v)(a)$. These two maps are well-defined as ${ker}(f)$ is a congruence. It is
routine to verify that the $K$-linear map $h=\varepsilon_f: V \ra W$ is an
$\cL$-homomorphism and $\mathrm{ker}(h)=\mathrm{ker}(\varepsilon_f)$, and
therefore $R=\mathrm{ker}(h)$.
\end{proof}

\begin{thm}
\label{thm:lwb}
Let $(V, \langle o,t \rangle)$ be a $K$-LWA and let $\llbracket \_
\rrbracket^\cL_V: V \ra K^{A^\ast}$ be the unique $\cL$-homomorphism into the final
coalgebra. Then $\mathrm{ker}(\llbracket \_ \rrbracket^\cL_V)$ is the largest
$K$-linear bisimulation on $V$.
\end{thm}

\begin{proof}
By Proposition~\ref{prop:lhomo}, $\mathrm{ker}(\llbracket \_ \rrbracket^\cL_V)$
is a $K$-linear bisimulation. Suppose $R$ is any $K$-linear bisimulation. Again by
Proposition~\ref{prop:lhomo}, there exists a $K$-LWA $(W,\langle o_W,t_W \rangle)$
and a $\cL$-homomorphism $f: V \ra W$ such that $R=\mathrm{ker}(f)$. Since
$(W,\langle o_W,t_W \rangle)$ is an $\cL$-coalgebra, there exists an
$\cL$-homomorphism $\llbracket \_ \rrbracket^\cL_W$ from $W$ to the final
coalgebra $K^{A^\ast}$. Therefore $\llbracket \_ \rrbracket^\cL_W \circ f: V \ra
K^{A^\ast}$, being the composition of two $\cL$-homomorphisms, is an
$\cL$-homomorphism. But $\mathrm{ker}\llbracket \_ \rrbracket^\cL_V$ is the
unique homomorphism from $V$ to $K^{A^\ast}$ by the finality of $K^{A^\ast}$ and
hence $\llbracket \_ \rrbracket^\cL_W \circ f = \mathrm{ker}\llbracket \_
\rrbracket^\cL_V$. Then $R = \mathrm{ker}(f) \subseteq \mathrm{ker}(\llbracket
\_ \rrbracket^\cL_W \circ f) = \mathrm{ker}(\llbracket \_ \rrbracket^\cL_V)$.
The set inclusion above is a consequence of the fact that $f(u)=f(v)$ implies
$g(f(u)=g(f(v))$ for all $g$ composable with $f$.
\end{proof}

\begin{cor}
\label{cor:lwb}
The $\cL$-behavioural equivalence relation $\approx_\cL$ on a $K$-LWA $V$ is the
largest $K$-linear bisimulation.
\end{cor}

\begin{proof}
By definition, $\approx_\cL = \{(v,w) \mid \llbracket v \rrbracket^\cL_V =
\llbracket w \rrbracket^\cL_V\} = \mathrm{ker}(\llbracket \_ \rrbracket^\cL_V)$.
The result follows from Theorem~\ref{thm:lwb}. \qedhere
\end{proof}

To summarize, for $K$-WA the largest $K$-weighted bisimulation $\sim_w$ is strictly
included in $K$-weighted language equivalence $\sim_l$ as shown in
Proposition~\ref{prop:refine}. Corollary~\ref{cor:lwb} shows that 
for $K$-LWA $K$-linear language equivalence 
coincides with the largest $K$-linear bisimulation. This raises the question:
what is the relationship between $K$-weighted bisimulation and $K$-linear
bisimulation? Again, our answer extends that of \cite{bonchi2012coalgebraic} to
the semimodule setting.

\begin{prop}
Let $(X, \langle o,t \rangle)$ be a $K$-WA and $(K(X), \langle o^\sharp,t^\sharp
\rangle)$ the $K$-LWA obtained from it as in Section~\ref{subsec:klwa}. If $R$
is a $K$-weighted bisimulation on $X$ then $R^\ell$ is a $K$-linear bisimulation
on $K(X)$.
\end{prop}

\begin{proof}
Recall the quotient weighted automaton $(X/R,\langle o_{X/R},t_{X/R}\rangle)$
and the map $\varepsilon_R: X \ra X/R$ from Section~\ref{sec:kwa}. From
diagram~\ref{diag:final} we have $\varepsilon_R$ is a $\cW$-homomorphism
between $(X, \langle o,t \rangle)$ and $(X/R,\langle o_{X/R},t_{X/R}\rangle)$.
Earlier we have shown that $K(h)$ is
an $\cL$-homomorphism for every $\cW$-homomorphism $h$. Therefore,
$K(\varepsilon_R): K(X) \ra K(X/R)$ is an $\cL$-homomorphism between
$(K(X), \langle o^\sharp,t^\sharp \rangle)$ and $(K(X/R), \langle o^\sharp_{X/R},
t^\sharp_{X/R} \rangle)$. By Proposition~\ref{prop:lhomo}, $\mathrm{ker}(K(\varepsilon_R)$
is a $K$-linear bisimulation on $K(X)$.

It remains to show that $\mathrm{ker}(K(\varepsilon_R) = R^\ell$. Recall that
$K(\varepsilon_R): K(X) \ra K(X/R)$ maps $k_i x_i$ to 
\[ (\sum_{x_j \in [x_i]_R} k_i)[x_i]_R \]
and hence by linearity maps any element $u = \sum_{i=1}^n k_i x_i$ of $K(X)$ to
\[ \sum_{i=1}^n (\sum_{x_j \in [x_i]_R}k_j)[x_i]_R. \]
Therefore, for any  element $v = \sum_{i=1}^ m k'_i y_i$  of $K(X)$,
$(u,v) \in \mathrm{ker}(K(\varepsilon_R))$ iff for all $x \in X$,
\begin{equation}
\label{eqn:lwb}
\sum_{x_j \in [x]_R} k_j = \sum_{y_\ell \in [x]_R} k'_\ell.
\end{equation}
We show that $(u,v) \in R^\ell$ if the same condition holds. First, since
$R$ is an equivalence relation, and $R^\ell$ is the smallest congruence containing $R$,
$R^\ell$ satisfies the following clauses (i) $xRy \Ra k.x R^\ell k.y$ for all $k \in K$ and
(ii) $xRy \mbox{ and } x'Ry' \Ra (k.x + k'.x') R^\ell (k.y + k'.y')$ for all $k_1,k_2 \in K$, and
(iii) for any $u,v \in K(X)$, $u R^\ell v$ only if it can be derived from the rules (i) and (ii) above.
It follows that for $u$ and $v$ as above, $u R^\ell v$ iff for all $x \in X$, 
Equation~\ref{eqn:lwb} holds.
\end{proof}

\section{$K$-Linear Partition Refinement}
\label{sec:lpr}

We have shown that weighted language equivalence $\sim_l$ coincides with the
largest $K$-linear bisimulation $\approx_\cL$ on a $K$-LWA. Now let us turn to
the question of computing $\sim_l$, \ie, given two states $v_1$ and $v_2$ in a
$K$-LWA $(V,\langle o,t\rangle)$ the problem of deciding whether $v_1 \sim_l
v_2$. For finite representability we assume that the submodule $V$ is freely generated
by a finite set $X$, and therefore $o$ and $t$ have finite representations as matrices.

Boreale~\cite{boreale2009weighted} and Bonchi \etal\
\cite{bonchi2012coalgebraic} had proposed two versions of a partition refinement
algorithm for finding the largest linear weighted bisimulation on a linear
weighted automaton $(V,\langle o,t\rangle)$. Their algorithms were based on the
identification of linear bisimulations with certain subspaces of the vector
space of a linear weighted automaton over a field. Both algorithms start from
the kernel of the output map $o$ represented as a subspace of $V$, and
successively obtain smaller and smaller subspaces by requiring invariance under
the transitions until a fixed point is reached. For termination the algorithm
relies on the fact that there can only be a finite descending chain of subspaces
for a finite-dimensional vector space. Moreover, the algorithm computes a
a basis for each of the subspaces by solving systems of linear equations, using
the matrix representations of the maps $o$ and $t$.

Unfortunately, these algorithms cannot, in general, be lifted to the
semiring-semimodule framework even for finitely generated semimodules. First,
unlike vector spaces, finitely generated semimodules are not necessarily
\emph{Artinian}. An Artinian semimodule is one that satisfies the descending
chain property, \ie, there is no infinite descending chain of subsemimodules
ordered by inclusion. For example, the semimodule $(\N,+,0)$ over the semiring
$(\N,+,.,0,1)$, although finitely generated, is \emph{not} Artinian since the
subsemimodules $2\N \supset 4\N \supset 8\N \ldots$ form an infinite descending
chain. Second, even with an Artinian semimodule the above procedures are not
effective in general, as they depend on solving linear equations in $K$, a problem that is
undecidable for certain semirings~\cite{narendran1996solving}. It is also known
that weighted language equivalence is undecidable for finite-state weighted
automata over the tropical semiring $(\N \cup \{\infty\},\min,+,\infty,0)$
~\cite{krob1994equality,almagor2020s}.

But all is not lost as far as partition refinement is concerned. We can
generalize the forward algorithm of \cite{boreale2009weighted} and
\cite{bonchi2012coalgebraic} to a construction of the final coalgebra in
$\mathsf{SMod}$ based on a method of of Ad{\'a}mek and
Koubek~\cite{adamek1995greatest} that uses the notion of the \emph{final
sequence} (or \emph{terminal sequence}). In general the construction takes
steps indexed by ordinals that can go beyond $\omega$, although for semimodules
no more than $\omega$ steps are necessary as is shown below, the proof being
identical to that of \cite{bonchi2012coalgebraic}. In special cases, such as for finite
dimensional vector spaces, and the free semimodules $\Z(X)$ and $\N(X)$
for a finite set $X$, the construction terminates after a finite number of
steps and the required operations can be performed effectively. In this case the
procedure is the same as the forward algorithm of
\cite{boreale2009weighted,bonchi2012coalgebraic}. We identify sufficient
conditions for the termination of the final sequence construction. The construction
terminates for Artinian semimodules as one would expect. But we also show that
a weaker condition identified in \cite{droste2012weighted,konig2018generalized}
also suffices. Our presentation below has a lot in common with that of
\cite{bonchi2012coalgebraic}, but also has major differences due to the
less enriched setting of semimodules as compared to vector spaces.

\subsection{The Partition Refinement Procedure}
\label{subsec:lpr}
The \emph{final sequence} of the functor $\cL = K \times (\_)^A: \mathsf{SMod}
\ra \mathsf{SMod}$ is the countable cochain $\{W_i\}_{i \in \N}$ shown below
\[ W_0=1 \stackrel{!}{\la} W_1=\cL (1) \stackrel{\cL(!)}{\la} W_2= \cL^2 (1) 
\stackrel{\cL^2 (!)}{\la} \ldots .\] 
Here $1$ is the terminal object $\{0\}$ in the category $\mathsf{SMod}$ and $!$
denotes the unique arrow from $\cL(1)$ to $1$. $\cL^i$ is simply the $i$-fold
composition of $\cL$, with $\cL^0$ being the the identity functor.

Ad{\'a}mek and Koubek~\cite{adamek1995greatest} showed that in a category $\cC$
with limits of all ordinal indexed cochains if the final sequence of a functor $T:
\cC \ra \cC$
\[ 1 \stackrel{!}{\la} T(1) \stackrel{T(!)}{\la} T^2 (1) 
\stackrel{T^2 (!)}{\la} \ldots .\] 
stabilizes at ordinal $k$, in the sense that the arrow $T^k(!) : T^{k+1}(1) \ra
T^{k}(1)$ in the final sequence is an isomorphism, then $(T^k(1),(T^k(!))^{-1})$ is
a final $T$-coalgebra. Moreover, the existence of a final coalgebra is
sufficient to ensure stabilization of the final sequence. Below we show that
in $\mathsf{SMod}$ the final sequence for $\cL$ stabilizes at or before the index $\omega$
by essentially repeating the proof for vector spaces in \cite{bonchi2012coalgebraic}.

The objects and arrows in the cochain can be described as follows. Following the
steps in \cite{bonchi2012coalgebraic} we can show that for each $n \in \N$,
$\cL^n(1)$ is isomorphic to $K^{A^\ast_n}$, where $A^\ast_n$ is the set of all
words of length less than $n$. The proof is by induction on $n$. For
$n=1$, $\cL(1)= K \times 1^A \cong K \cong K^{\{\epsilon\}} = K^{A^\ast_1}$ by
definition. By the induction hypothesis, we have that an element $\langle
k,\sigma \rangle \in \cL^{n+1}(1) = K \times \cL^n(1)^A \cong K \times
(K^{A^\ast_n})^A \cong K^{1+A\times A^\ast_n}$ can be seen as a function that
maps $\epsilon$ to $k$ and $aw$ to $\sigma(a)(w)$ for $a \in A$ and $w \in
A^\ast_n$. Similarly, it can be shown that the arrow $\cL^n (!): \cL^{n+1}(1)
\ra \cL^n(1)$ maps a function $\sigma \in K^{A^\ast_{n+1}}$ to its restriction
$\sigma \restriction A^\ast_n$ to the subset of words of length less than $n$.
We denote $\sigma \restriction A^\ast_n$ by $\sigma \restriction n$ from now on.
Moreover, the limit of this $\mathsf{SMod}$-cochain is $K^{A^\ast}$ together
with the maps $\zeta_n : K^{A^\ast} \ra \cL^n(1)$ that send a weighted language
$\sigma: A^\ast \ra K$ to its restriction $\sigma \restriction n$ for each $n$.
The limit cone is shown in the commuting diagram below.

\begin{equation}
\begin{tikzcd}
& & & K^{A^\ast} \arrow[dlll, bend right=20, "\zeta_0"']
	\arrow[dll, bend right=15, "\zeta_1"']
	\arrow[dl, bend right=10, "\zeta_2"'] \\
1 \arrow[leftarrow,"!"]{r} & \cL (1) \arrow[leftarrow,"\cL(!)"]{r}
& \cL^2 (1) \arrow[leftarrow,"\cL^2 (!)"]{r} & \ldots
\end{tikzcd}
\end{equation}

On the other hand, any $K$-LWA $(V,\langle o,t \rangle)$, where $V$ is not
necessarily an Artinian semimodule, defines a cone $!_{n}: V \ra \cL^n(1)$
by induction as follows. The map $!_0: V \ra 1$ is the unique arrow to
the terminal object $1$, and $!_n = \cL(!_n) \circ \langle o,t \rangle$. More
concretely, for all $w \in A^\ast$,
\begin{equation}
\label{eqn:shriek}
!_{n+1}(v)(w) =
\begin{cases}
o(v),\myquad[5]\,\,\, &\mbox{if } w = \epsilon\\
!_n(t(v)(a)(w')) \myquad[2] &\mbox{if } w = aw'.
\end{cases}
\end{equation}

Equation~\ref{eqn:shriek} says $!_n(v)(w) = \sigma_l(v)(w)$ for all $w$ with
$|w| \leq n$, and $!_n(v)(w) = \sigma_l(v)(w|_n)$ for all $w$ with $|w| > n$,
where $w|_n$ is the prefix of $w$ of length $n$.

It is easily verified that the unique arrow from the cone from $V$ to the limit cone
from $K^{A^\ast}$ is the same as the unique $\cL$-coalgebra homomorphism $\llbracket
\_ \rrbracket^\cL_V$ from $(V,\langle o,t \rangle)$ to the final coalgebra
($K^{A^\ast},\epsilon,d)$. This is depicted in the following commuting diagram.

\begin{equation}
\begin{tikzcd}
& & & K^{A^\ast} \arrow[dlll, bend right=20, "\zeta_0"']
	\arrow[dll, bend right=15, "\zeta_1"']
	\arrow[dl, bend right=10, "\zeta_2"'] \\
1 \arrow[leftarrow,"!"]{r} & \cL (1) \arrow[leftarrow,"\cL(!)"]{r}
& \cL^2 (1) \arrow[leftarrow,"\cL^2 (!)"]{r} & \ldots \\
& & & V \arrow[ulll, bend left=20, "!_0"']
	\arrow[ull, bend left=15, "!_1"']
	\arrow[ul, bend left=10, "!_2"']
	\arrow[uu, bend right=40, "{\llbracket
\_ \rrbracket^\cL_V}"']
\end{tikzcd}
\end{equation}

Recall from Theorem~\ref{thm:lwb} and Corollary~\ref{cor:lwb} that the
$\cL$-behavioural equivalence relation $\approx_\cL$ on an $(V,\langle o,t
\rangle)$ is the kernel of $\llbracket \_ \rrbracket^\cL_V$. An abstract
procedure for computing the equivalence $\approx_\cL$ would be to iteratively
compute the kernel of the arrows $!_n$ and terminate (if ever) when
$\mathrm{ker}(!_{n+1})=\mathrm{ker}(!_n)$. This condition is equivalent to
$\cL^n(!): \cL^{n+1}(1) \ra \cL^n(1)$ being an isomorphism. But first we prove two
propositions that justify the abstract procedure. The second result requires $V$ to be
a finitely generated Artinian semimodule.

\begin{prop}
\label{prop:pr-forward}
Let $(V,\langle o,t \rangle)$ be a $K$-LWA. Consider the sequence of relations over
$V$ defined inductively by
\[ R_0 = \mathrm{ker}(o) \myquad[2] R_{i+1} = R_i \cap \bigcap_{a \in A} \{(u,v) \mid
(t(u)(a),t(v)(a)) \in R_i \}. \] 
Then for all $n$, $R_n$ is a linear relation satisfying $R_n =
\mathrm{ker}(!_{n+1})$.
\end{prop}

\begin{proof}
We show that $R_n = \mathrm{ker}(!_{n+1})$ for each $n$ by induction. It
follows that $R_n$ is a linear relation. For the basis, since $!_1 = o$ we have
$R_0 = \mathrm{ker}(o) = \mathrm{ker}(!_1)$. Suppose $R_n=
\mathrm{ker}(!_{n+1})$ for $n \geq 1$. Then
\begin{align*}
R_{n+1} &= R_n \cap  \bigcap_{a \in A} \{(u,v) \mid (t(u)(a),t(v)(a)) \in R_n \} \\
&= \mathrm{ker}(!_{n+1}) \cap  \bigcap_{a \in A} \{(u,v) \mid (t(u)(a),t(v)(a)) \in {ker}(!_{n+1}) \} \\
&= \mathrm{ker}(!_{n+1}) \cap  \bigcap_{a \in A} \mathrm{ker}(!_{n+1} \circ \langle t(\_)(a), t(\_)(a)
\rangle) \\
&= \mathrm{ker}(\langle !_{n+2}) \mbox{ from Equation~\ref{eqn:shriek}  }
\end{align*}
Here we use the following facts that can be easily checked.
\be
\item[(1)] $ \mathrm{ker}(g \circ f) = \{(u,v) \mid (f(u),f(v)) \in
\mathrm{ker}(g)\}$ and
\item[(2)] $\mathrm{ker}(\langle f_1,f_2 \rangle) = \mathrm{ker}(f_1) \cap
\mathrm{ker}(f_2)$. \qedhere
\ee
\end{proof}

It follows from Proposition~\ref{prop:pr-forward} that when $V$ is a finite
dimensional vector space then for some $n \in \N$,
$\mathrm{ker}(!_{n+1})=\mathrm{ker}(!_n)$ \ie, the procedure terminates after a
finite number of steps. This is because for vector spaces and linear maps,
kernels (in the sense of \cite{bonchi2012coalgebraic}) correspond to subspaces
and for a finite dimensional space there can only be a finite chain of subspace
containment; see \cite{bonchi2012coalgebraic} for the details. For the more
general case of finitely generated Artinian semimodules over a semiring, we can
also guarantee stabilization of the cochain. This is shown in the following
proposition.

\begin{prop}
\label{prop:stabilization}
Let $(V,\langle o,t \rangle)$ be a $K$-LWA where $V$ is a finitely generated Artinian
semimodule. Consider the sequence of relations over $V$ defined inductively by
\[ R_0 = \mathrm{ker}(o) \myquad[2] R_{i+1} = R_i \cap \bigcap_{a \in A} \{(u,v) \mid
(t(u)(a),t(v)(a)) \in R_i \}. \] 
Then there is a $j$ such that $R_j = R_{j+1}$. The largest $K$-linear
bisimulation $\approx_\cL$ on $V$ is then identical to $R_j$.
\end{prop}

\begin{proof}
We have shown that $R_n = \mathrm{ker}(!_{n+1})$ for each $n$. Since each $R_i$
is linear it is a subsemimodule of the product $V \times V$. Since $V$ is
Artinian so is $V\times V$ as the projection of a subsemimodule of $V\times V$ is
a subsemimodule of $V$. Since $R_i \supseteq R_{i+1}$, and we cannot have an infinite
descending chain of subsemimodule inclusions there exists a $j$ such that $R_j =
R_{j+1}$.

We now show that $R_j$ is a $K$-linear bisimulation by applying
Lemma~\ref{lem:kbisim}. Since $R_j \subseteq R_0=\mathrm{ker}(o)$, condition (1)
of the lemma is satisfied. Since $R_j=R_{j+1}$, we have $R_j = R_j \cap
\bigcap_{a \in A} \{(u,v) \mid (t(u)(a),t(v)(a)) \in R_j \}$, \ie, $uR_jv$
implies $t_auR_jt_av$, for each $a \in A$, which means condition (2) is also
satisfied.

Finally, we must show that any $K$-linear bisimulation $R$ is included in $R_j$. We
do this by showing $R \subseteq R_i$ for all $i$ by induction. By
Lemma~\ref{lem:kbisim}, $R \subseteq \mathrm{ker}(o) = R_0$. Now assume $R
\subseteq R_i$. By Lemma~\ref{lem:kbisim} again, $uRv$ implies $t_auRt_av$ for
all $a \in A$ and hence, $t_auR_it_av$ for all $a \in A$. By definition of
$R_{i+1}$ it follows that $uR_{i+1}v$ and hence $R \subseteq R_{i+1}$.
\end{proof}

We now show that the termination condition in
Proposition~\ref{prop:stabilization}, namely $V$ is Artinian, can be weakened.
This weaker condition is mentioned in
\cite{droste2012weighted,konig2018generalized} but in a somewhat different
setting. To state the condition, let $V_n \subseteq V = K(X)$ for a finite set
$X$ be defined by $V_n = \{\sigma_l(\_)(w): X \ra K \mid w \in A^\ast_n\}$ and
let $V_\ast = \bigcup_{n=0}^\infty V_n$. Then the following proposition states
that the procedure described in Proposition~\ref{prop:stabilization} terminates
if $\mathrm{span}(V_\ast)$ is a finitely generated semimodule.

\begin{prop}
\label{prop:weak-stab}
Let $(V,\langle o,t \rangle)$ be a $K$-LWA where $V=K(X)$ for a finite set $X$,
such that $\mathrm{span}(V_\ast)$ is finitely generated. Consider the sequence
of relations over $V$ defined inductively by
\[ R_0 = \mathrm{ker}(o) \myquad[2] R_{i+1} = R_i \cap \bigcap_{a \in A} \{(u,v) \mid
(t(u)(a),t(v)(a)) \in R_i \}. \] 
Then there is a $j$ such that $R_j = R_{j+1}$. The largest $K$-linear
bisimulation $\approx_\cL$ on $V$ is then identical to $R_j$.
\end{prop}

\begin{proof}
Suppose $\mathrm{span}(V_\ast)$ is finitely generated. Since $V_n \subseteq
V_{n+1}$ for all $n$, and $V_\ast = \bigcup_{n=0}^\infty V_n$, we have, for some
$j$, $\mathrm{span}(V_\ast) = \mathrm{span}(V_j) = \mathrm{span}(V_{j+1})$.
Below we show that $\mathrm{ker}(!_{j}) = \mathrm{ker}(!_{j+1})$. It follows
from Proposition~\ref{prop:pr-forward} that $R_j = R_{j+1}$ The rest of the
proof is identical to the one for Proposition~\ref{prop:stabilization}.

It is immediate from Proposition~\ref{prop:pr-forward} that
$\mathrm{ker}(!_{j+1}) \subseteq \mathrm{ker}(!_{j})$, as $R_{j+1} \subseteq
R_j$. To prove the containment in the other direction, suppose $(u,v) \in
\mathrm{ker}(!_j)$. Then $\sigma_l(u)(w)= \sigma_l(v)(w)$ for all $w \in
A^\ast_j$ from Equation~\ref{eqn:shriek}. Now, since $\mathrm{span}(V_j) =
\mathrm{span}(V_{j+1})$, any element of $V_{j+1}$ is a linear combination
of elements of $V_j$, \ie, for any $w' \in A^\ast_{j+1}$ we have,
$\sigma_l(\_)(w') = \sum_{i=1}^n \sigma_l(\_)(w_i)$ for some $w_i \in A^\ast_j$,
$1 \leq i \leq n$. Therefore,
\begin{align*}
\sigma_l(u)(w') &= \sum_{i=1}^n \sigma_l(u)(w_i) \text{ for some } w_i \in A^\ast_j,
1 \leq i \leq n\\
&= \sum_{i=1}^n \sigma_l(v)(w_i) \text{ since  } \sigma_l(u)(w)=
\sigma_l(v)(w) \text{ for all } w \in A^\ast_j\\
&= \sigma_l(v)(w'),
\end{align*}
\ie, $(u,v) \in \mathrm{ker}(!_{j+1})$.
\end{proof}

However, even when the final sequence stabilizes after a finite number of steps, in general
there is no effective procedure to compute the kernel of a $K$-linear map for a semiring
$K$. This is because unlike in $\R$, $\Z$ and $\N$, for semirings in general we do not have a procedure for linear equation solving; the problem is undecidable for certain semirings~\cite{narendran1996solving}. We conclude this section by presenting
a couple of examples from different semimodules to illustrate how our procedure works.

\begin{figure}[t]
\begin{center}
\caption{A $K$-linear weighted automaton over the ring $\Z$}
\vspace{0.5cm}
\label{fig:lwa-1}
\begin{adjustbox}{width=3cm}
\begin{tikzpicture}
\node[state] (q0) {$x_0$}; 
\node[state,below= of q0, label=right:$1$] (q1) {$x_1$}; 
\draw (q0) edge[left] node{$a,1$} (q1);
\draw (q1) edge [loop below] node{$a,2$} (q1);
\node[state,right= of q0] (q2) {$x_2$}; 
\node[state,below = of q2, label=right:$1$] (q3) {$x_3$}; 
\node[state,below= of q3, label=right:$1$] (q4) {$x_4$}; 
\draw (q2) edge[right] node{$a,1$} (q3);
\draw (q3) edge[right] node{$a,1$} (q4);
\draw (q4) edge [loop below] node{$a,2$} (q4);
\end{tikzpicture}
\end{adjustbox} 
\end{center}
\end{figure}
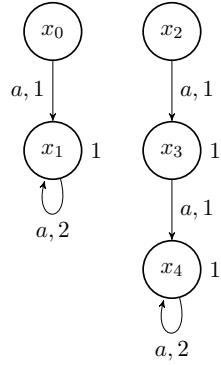

\begin{exa}
\label{ex:lwa-1}
Figure~\ref{fig:lwa-1} shows a $K$-LWA $(V,\langle o,t \rangle)$, adapted from
\cite{Sakarovitch2016slides}, over the alphabet $A=\{a\}$. The semiring $K$ is
the ring $\Z$ of integers. Here $V = \Z(X)$ where $X=\{x_0,x_1,x_2,x_3,x_4\}$.
In the example, the maps $o$ and $t$ are generated from their values shown
against the nodes and edges, respectively, by linear extension. Missing values
against nodes are assumed to be zero. Note that $V$ is \emph{not} Artinian but
$\mathrm{span}(V_\ast)$ is indeed finitely generated. To see this, $V_3$
consists of the set containing the following maps from $X$ to $\Z$
\begin{align*}
\sigma_l(\_)(\epsilon) = \{x_0 \mapsto 0, x_1 \mapsto 1, x_2 \mapsto 0, x_3 \mapsto 1, x_4 \mapsto 1\} \\
\sigma_l(\_)(a) = \{x_0 \mapsto 1, x_1 \mapsto 2, x_2 \mapsto 1, x_3 \mapsto 1, x_4 \mapsto 2\} \\
\sigma_l(\_)(aa) = \{x_0 \mapsto 2, x_1 \mapsto 4, x_2 \mapsto 1, x_3 \mapsto 2, x_4 \mapsto 4\}.
\end{align*}
Now, it is easily seen that for all $x \in X$ and $n \geq 2$, $\sigma_l(x)(a^n) =
2^{n-2}\sigma_l(x)(aa)$. Therefore, $\mathrm{span}(V_\ast)$ is generated by
$V_3$. To apply the partition refinement procedure to this $K$-LWA, we use
Propositions~\ref{prop:pr-forward} and \ref{prop:weak-stab}. We have
\begin{align*}
R_0 = \mathrm{ker}(o) = \{(u,v) \mid & u = \sum_{i=0}^4 c_i x_i,
v = \sum_{i=0}^4 d_i x_i, \text{ where } c_i,d_i \in \Z \text{ for } i \in [0,4]\\
& \text{ and }\\
&c_1+c_3+c_4 = d_1 + d_3 + d_4 \},
\end{align*}
as $x_1$ , $x_3$ and $x_4$ are the only states with nonzero weights and all
of them have weight one. In the next iteration, we have
\begin{align*}
R_1 = R_0 \cap \{(u,v) \mid \: & u = \sum_{i=0}^4 c_i x_i,
v = \sum_{i=0}^4 d_i x_i, \text{ where } c_i,d_i \in \Z \text{ for } i \in [1,4]\\
& \text{ and } \\
& c_0 + 2c_1 + c_2 + c_3 + 2c_4 = d_0 + 2d_1 + d_2 + d_3 + 2d_4\},
\end{align*}
as all the edges have weight one except the loops on $x_1$ and $x_4$, which have
weight two. It is clear that $R_2= R_1$, since we take the intersection with the
same set in obtaining $R_1$ from $R_0$ as in obtaining $R_2$ from $R_1$.
Therefore, the $K$-linear bisimilarity relation, \ie, weighted language
equivalence, is given by the set of all pairs $(u,v)$ with $u = \sum_{i=0}^4 c_i
x_i$, $v = \sum_{i=0}^4 d_i x_i$ satisfying the integer equations
\begin{align*}
&c_1+c_3+c_4 = d_1 + d_3 + d_4 \\
&c_0 + 2c_1 + c_2 + c_3 + 2c_4 = d_0 + 2d_1 + d_2 + d_3 + 2d_4.
\end{align*}
\end{exa}

When $V = \Z(X)$ for a finite set $X$, it can be shown that the number of
iterations of the loop computing $R_{i+1}$ from $R_i$ is bounded by $|X|$ and
termination then follows from the decidability of linear integer arithmetic. A
more efficient algorithm for deciding language equivalence between states in $X$
(and not the entire $V(X)$ as in our case) appears in
\cite{beal2005equivalence}. In fact, this problem is decidable for any semiring
that is a subsemiring of a field~\cite{sakarovitch2009rational}.

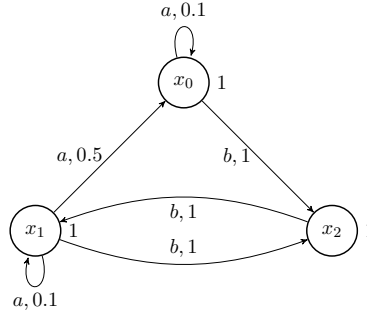
\begin{figure}[t]
\begin{center}
\caption{A $K$-linear weighted automaton over the semiring $K=([0,1],\max,.,0,1)$}
\vspace{0.5cm}
\label{fig:lwa-1}
\begin{adjustbox}{width=5cm}
\begin{tikzpicture}[node distance=3cm]
\node[state, label=right:$1$] (q0) {$x_0$}; 
\node[state,below left = of q0, label=right:$1$] (q1) {$x_1$}; 
\node[state,below right = of q0, label=right:$1$] (q2) {$x_2$}; 
\draw (q0) edge[left] node{$b,1$} (q2);
\draw (q0) edge [loop above] node{$a,0.1$} (q0);
\draw (q1) edge[left] node{$a,0.5$} (q0);
\draw (q1) edge[bend right=20, above] node{$b,1$} (q2);
\draw (q2) edge[bend right=20, below] node{$b,1$} (q1);
\draw (q1) edge [loop below] node{$a,0.1$} (q1);
\end{tikzpicture}
\end{adjustbox} 
\end{center}
\end{figure}

\begin{exa}
\label{ex:lwa-2}
Figure~\ref{fig:lwa-1} shows a $K$-LWA $(V,\langle o,t \rangle)$, adapted from
\cite{konig2018generalized}, over the alphabet $A=\{a,b\}$ and semiring
$K=([0,1],\max,.,0,1)$. Here $V = K(X)$ where $X=\{x_0,x_1,x_2\}$. It can be
checked that $V$ is not an Artinian semimodule but $V_\ast$ is finitely
generated. As in the example above, we use Propositions~\ref{prop:pr-forward}
and \ref{prop:weak-stab} to apply the partition refinement procedure. We have
\begin{align*}
R_0 = \mathrm{ker}(o) = \{(u,v) \mid \: & u = \sum_{i=0}^2 c_i x_i,
v = \sum_{i=0}^2 d_i x_i, \text{ where } c_i,d_i \in [0,1] \text{ for } i \in \{0,1,2\}\\
& \text{ and }\\
& \max\{c_0,c_1,c_2\} = \max\{d_0,d_1,d_2 \} \}.
\end{align*}
In the next iteration we have,
\begin{align*}
R_1 = R_0 \cap \{(u,v) \mid \: & u = \sum_{i=0}^2 c_i x_i,
v = \sum_{i=0}^2 d_i x_i, \text{ where } c_i,d_i \in [0,1] \text{ for } i \in \{0,1,2\}\\
& \text{ and } (t(u)(a),t(v)(a)) \in R_0 \\
&  \text{ and } (t(u)(b),t(v)(b)) \in R_0\}.
\end{align*}

In the above expression, $t(u)(a) = \max \{0.1c_0,0.5c_1,0.1c_1\}$, $t(v)(a) =
\max \{0.1d_0,0.5d_1,0.1d_1\}$, $t(u)(b) = \max \{c_1,c_2\}$ and $t(v)(b) =
\max \{d_1,d_2\}$. Again, we can check that $R_2=R_1$. The resulting linear
equations over $K$ can be solved by using the theory of $l$-monoids and residuation~\cite{konig2018generalized}.
\end{exa}

\subsection{Related Partition Refinement Algorithms}
\label{subsec:konig}
For specific semirings, algorithms have been
proposed in the case of probabilities~\cite{kiefer2011language},
fields~\cite{boreale2009weighted}, rings~\cite{droste2012weighted}, division
rings~\cite{flouret1997noncommutative} and principal ideal
domains~\cite{beal2005equivalence}. 

The approach to partition refinement which has some similarity with ours is by
K{\"o}nig and K{\"u}pper~\cite{konig2018generalized}. This section is a brief
summary of their work followed by a comparison with the current paper. The paper
generalizes the partition refinement algorithm to a coalgebraic setting, just as
this paper. It proposes a generic procedure for language equivalence for
transition systems, of which weighted automata over arbitrary semirings form an
important special case. However, unlike the usual partition refinement approach,
\cite{konig2018generalized} does not provide a unique or canonical
representative for the weighted language accepted, in general. The procedure
based on the approach is not guaranteed to terminate in all cases, but does so
for particular semirings, just as in our case. It is based on a notion of
equivalence classes of arrows and involves solving linear equations for a given
semiring.

In more detail, \cite{konig2018generalized} uses $\cM(K)$, the category of
finite sets and matrices over the semiring $K$
with matrix multiplication as composition, as the base category. This is just the
Kleisli category of the free semimodule monad and is equivalent to the category
of free semimodules over finite sets that we use. The endofunctor $F$ over
$\cM(K)$ defining a weighted automaton is given by $FX = 1 + A \times X$.
This can be seen as being equivalent to our $\cL$, as the paper
\cite{konig2018generalized} uses matrices as arrows rather than maps.

The paper presents two generic procedures for partition refinement, where the
second is just an optimized version of the first. The basic ingredient in the
theory which is new is a preorder on objects and arrows of a concrete category
$\cC$. For objects $X$ and $Y$ in $\cC$, this is defined by $X \leq Y$ if there
is an arrow $f : X \ra Y$. The relation $X \equiv Y$ holds when $X \leq Y$ and
$Y \leq X$. For arrows $f: X \ra Y$ and $g: X \ra Z$ with the same domain, $f
\leq^X g$ if there exists an arrow $h: Y \ra Z$ such that $g = h \circ f$.
Similarly, $f \equiv^X g$ if $f \leq^X g$ and $g \leq^X f$. It is easy to check
that $\leq$ and $\leq^X$ are preorders and $\equiv$ and $\equiv^X$ are
equivalence relations.

Given a coalgebra $\alpha: X \ra FX$ in $\cC$ and an arrow $f: X \ra Y$, $f$ is
called a \emph{postfixed point} if $f \leq^X Ff \circ \alpha$. This is
equivalent to the existence of a coalgebra $\beta: Y \ra FY$ with $f$ an
$F$-homomorphism by simply taking $\beta$ as the mediating morphism that
witnesses $f \leq^X Ff \circ \alpha$. The paper then shows
that the largest postfixed point $f: X \ra Y$ induces behavioural equivalence:
$x \approx_F y$ if and only if $f(x)=f(y)$.

The first partition refinement procedure in \cite{konig2018generalized}, called
Procedure A, uses the final sequence construction to find the largest postfixed
point for general coalgebras. The second, called Procedure B, is an optimized
version intended for weighted automata over a semiring. Both procedures take a
coalgebra $\alpha: X \ra FX$ as input, and if they terminate, return a coalgebra
$\beta: Y \ra FY$ along with an $F$-homomorphism $f: X \ra Y$ from $\alpha$ to
$\beta$ that induces behavioural equivalence on $X$: $x \approx_\cF y$ if and
only if $f(x)=f(y)$. Here $x$ and $y$ are elements of the underlying set of $X$.
Note that $\beta$ need not be the final coalgebra in general.

Procedure $A$ constructs the arrows $!_n: X \ra F^n(1)$ in the diagram below
\begin{equation}
\begin{tikzcd}
%& & & K^{A^\ast} \arrow[dlll, bend right=20, "\zeta_0"']
%	\arrow[dll, bend right=15, "\zeta_1"']
%	\arrow[dl, bend right=10, "\zeta_2"'] \\
1 \arrow[leftarrow,"!"]{r} & F (1) \arrow[leftarrow,"F(!)"]{r}
& F^2 (1) \arrow[leftarrow,"F^2 (!)"]{r} & \ldots \\
& & & X \arrow[ulll, bend left=20, "!_0"']
	\arrow[ull, bend left=15, "!_1"']
	\arrow[ul, bend left=10, "!_2"']
%	\arrow[uu, bend right=40, "{\llbracket
%\_ \rrbracket^\cL_V}"']
\end{tikzcd}
\end{equation}
by induction as follows: $!_0: X \ra 1$ is the unique arrow to the terminal
object and $!_{n+1} = F(!_n) \circ \alpha : X \ra F^{n+1}(1)$. What is novel
here is the termination condition for Procedure A, namely $!_n \leq^X !_{n+1}$,
\ie, there is an arrow $\beta: F^n(1) \ra F^{n+1}(1)$ such that $\beta \circ !_n
= !_{n+1}$. This condition is equivalent to $!_n \equiv^X !_{n+1}$, which is
weaker than requiring that $F^n(!)$ is an isomorphism, \ie, the
elements of the cochain stabilize at $n$, as is the case for our procedure.
This can lead to earlier termination for Procedure A. When applied to weighted
automata over a semiring $K$, the procedure yields for each $n$, $!_{n+1}(x)(w)
= \sigma_l(x)(w)$ for $w \in A^\ast_n$, just as in our case; see
Equation~\ref{eqn:shriek}. The termination condition then reduces to checking
whether the semimodule $\mathrm{span}(\{\sigma_l(\_)(w) \mid w \in
A^\ast_{n+1}\})$ equals $\mathrm{span}(\{\sigma_l(\_)(w) \mid w \in
A^\ast_n\})$, the condition in Proposition~\ref{prop:weak-stab}, which can be
done by solving linear equations in the semiring $K$.

Procedure B replaces $!_n$ for each $n$ by any $e_n$ such that $!_n \equiv^X
e_n$, thus potentially reducing the search space. The termination depends on the
semiring. In particular, termination is guaranteed when the semimodule $K(X)$ is
Artinian, as well as when the weaker condition mentioned in
Proposition~\ref{prop:weak-stab} holds. To summarize, although the termination
condition for the procedures in \cite{konig2018generalized} is different from ours and it
can lead to more efficiency in principle, the underlying framework of the final
chain is identical to ours. The question whether there are examples where Procedure A
(or B) terminates and our procedure does not is still open.

\section{Conclusion}

In this paper we have continued the programme in Bonchi
\etal~\cite{bonchi2012coalgebraic} of a coalgebraic theory of bisimulation of
weighted automata. We have generalized the formal framework to one based on
weighted automata on semirings rather than fields. As we have shown, almost all
the results of \cite{bonchi2012coalgebraic} continue to hold except for the
partition refinement algorithm. The latter necessarily depends on the specific
type of semiring involved. Hence we can only propose an abstract procedure for
partition refinement that may not, in general, halt or lead to a solution even
if it halts, because of the the lack of enough operations. We have provided
sufficient conditions for the procedure to halt. We have also compared our work
with that in \cite{konig2018generalized}, which provides another partition
refinement algorithm for weighted automata over semirings in the coalgebraic
framework.

\section*{Acknowledgment}
The author thanks David Benson for introducing him to category theory
35 years ago. He also wishes to acknowledge the anonymous referees
for their help in substantially refining the contents of Section~\ref{sec:lpr} of
the paper on partition refinement.

\bibliographystyle{alphaurl}
\bibliography{bisim.bib}    
\end{document}